\documentclass[%
 reprint,
 onecolumn,
 amsmath,
 amssymb,
 aps,
 unsortedaddress
]{revtex4-2}

\usepackage[margin=1in]{geometry}
\usepackage{amsmath, amssymb, amsthm, mathtools}
\usepackage{hyperref}
\usepackage{enumitem}
\usepackage[svgnames]{xcolor}

\hypersetup{
    colorlinks=true,
    linkcolor=blue,
    citecolor=blue,
    urlcolor=blue
}

\newtheorem{theorem}{Theorem}[section]
\newtheorem{lemma}[theorem]{Lemma}
\newtheorem{proposition}[theorem]{Proposition}
\newtheorem{corollary}[theorem]{Corollary}
\theoremstyle{definition}
\newtheorem{definition}[theorem]{Definition}
\newtheorem{example}[theorem]{Example}
\theoremstyle{remark}
\newtheorem{remark}[theorem]{Remark}

\newcommand{\C}{\mathbb{C}}
\newcommand{\Z}{\mathbb{Z}}
\newcommand{\T}{\mathbb{T}}
\newcommand{\N}{\mathbb{N}}

\newcommand{\R}{\mathbb{R}}
\newcommand{\ii}{\mathrm{i}}
\newcommand{\dd}{\mathrm{d}}

\newcommand{\dist}{\operatorname{dist}}

\newcommand{\ket}[1]{\left|#1\right\rangle}
\newcommand{\bra}[1]{\left\langle#1\right|}

\newcommand{\norm}[1]{\left\lVert #1 \right\rVert}
\newcommand{\abs}[1]{\left|#1\right|}

\begin{document}

\title{Moment-Structured Block Encodings of Periodic Finite-Difference Operators}
\author{Jishnu Mahmud}
\email{jmahmud@vols.utk.edu}
%\affiliation{Department of Industrial and Systems Engineering, University of Tennessee Knoxville, USA}
\author{Rebekah Herrman%\orcidlink{0000-0001-6944-4206}}\thanks{corresponding author
}
\email{rherrma2@utk.edu}
\affiliation{Department of Industrial and Systems Engineering, University of Tennessee Knoxville, USA}
\date{\today}% It is always \today, today,

\begin{abstract}

Block encoding is the standard technique for accessing matrix data in quantum linear-algebra algorithms. Its implementation directly affects its subnormalization, which in turn controls the algorithm's success probability, simulation time, and downstream costs. Explicit construction of block encodings with provably optimal subnormalization exists for only a handful of operators, with bespoke calculations used in its design. In this work, we develop a framework to block encode translation-invariant finite-difference operators on a periodic grid. These operators are the finite-difference discretizations of the constant-coefficient partial differential equations that sit at the core of scientific computing. We show that the moment order of these stencils can be used to simultaneously determine the continuum operator approximated, the vanishing order of the Fourier symbol, and the cost of the block encoding. From there, we derive a closed-form optimality criterion as a function of the stencil coefficients, which certifies whether the construction attains the optimal subnormalization for an entire operator family, uniformly in grid size, and quantifies the gap when it does not. The framework subsumes optimal constructions for the Laplacian operator in the literature and can be used to certify new instances at higher even orders, including the biharmonic operator. Furthermore, we derive success-probability floors parameterized by spectral properties of the operator's symbol and find explicit constants for the block encoding of the advection-diffusion family for which no prior explicit spatial block encoding exists.

\end{abstract}
\maketitle

\section{Introduction}
\label{sec:intro}
Partial differential equations (PDEs) are foundational models in scientific computing, appearing across fluid dynamics, materials science, transport, wave propagation, and many-body simulation \cite{antontsev2002energy, bates2006some, ikawa2000hyperbolic, von1999multibody}. The prospect of solving them using quantum computing has spurred extensive research. For example, quantum linear systems and differential-equation algorithms based on the quantum singular value transformation (QSVT) \cite{gilyen2019quantum, martyn2021grand} have made the discretized differential operator central to their work. For any of these algorithms to function, the first step is to implement the non-unitary action as a unitary operator that can then act on the data. This implementation is carried out using a method known as block-encoding, in which the target operator $A$, in its matrix form, normalized by a subnormalization factor $\lambda$, is embedded in the top-left corner of a larger unitary matrix \cite{low2019hamiltonian, gilyen2019quantum}. This subnormalization constant is of immense importance, as it not only directly scales the success probability when the operator action is recovered after readout but also scales Hamiltonian simulation times and multiplies the condition number seen by a downstream primitive such as QSVT.

There are many block encoding methods, each with its own stark tradeoff. Earlier works \cite{low2019hamiltonian, gilyen2019quantum, berry2007efficient, berry2015hamiltonian, childs2017quantum} design encoders that are operator-agnostic and rely on Quantum Random Access Memory or other black-box oracle access, which are often more difficult to implement than the original problem. These methods can encode any arbitrary matrix but can require a prohibitive number of quantum gates or circuit depth. Furthermore, they yield a subnormalization that scales with the matrix dimension rather than its spectral norm. It is apparent in recent years that such generalized block encoding exhibits extremely unfavorable cost scaling, making it infeasible unless full fault tolerance is achieved; therefore, efforts must resort to the exploitable structure of the specific operator \cite{kuklinski2025efficient}.

Recent work has therefore focused on developing structure-exploiting block-encoding methods with explicit representations. This research landscape appears to be split into two tracks. The first constructs explicit circuits for classes of operators that share structural properties worth exploiting: certain sparse matrices \cite{camps2024explicit}, Toeplitz and banded matrices \cite{sunderhauf2024block}, sparse matrices with periodic diagonal structure \cite{zecchi2026block}, and, at the framework level, pseudo-differential operators organized by the separability of their symbols between position and frequency, as well as by dimension-wise separability \cite{li2023efficient}. These works have the definite objective of generalizing over a class, but they do not provide a class-wide certificate that the subnormalization attained by the block-encoding construction is optimal. Setty in \cite{setty2025block} identified matching the subnormalization to the operator's spectral norm as an open problem for general sparse-matrix constructions. 

The second line of research studies optimal block encodings for more specialized operators. The discrete Laplacian under periodic boundary conditions achieves optimal subnormalization in \cite{sturm2025efficient}, under the conditions established there. Broader explicit constructions for different boundary conditions and elliptic operators are developed in \cite{kharazi2025explicit}, although without a corresponding optimality result for the entire class. In this line of work, the optimality claim is established by explicitly computing or otherwise characterizing the relevant eigenvalues, so the optimality argument does not automatically transfer to a new operator without a fresh spectral analysis. This leads us to observe that class-level constructions generally lack class-wide optimality guarantees, while existing optimality proofs have more limited class-level reach.

In this paper, we work at the intersection of the two tracks. We consider the class of translation-invariant finite-difference operators on periodic grids. This wide and highly applicable class of operators does not seem to have an exploitable property immediately apparent in prior work. For example, the symbol of a translation-invariant operator depends only on frequency, so the position-frequency separability used in \cite{li2023efficient} does not distinguish operators within this subclass. We observe, however, that these operators are linear combinations of lattice shifts and therefore admit an explicit block encoding whose subnormalization equals the sum of the coefficient weights. They are also diagonal in the Fourier basis. Together, these observations allow us to derive an analytic criterion based on the operator's spectral norm and thereby certify the minimum possible subnormalization across the entire class. Although the general case remains open, our work resolves it analytically for translation-invariant shift-stencil operators under a verifiable phase-alignment condition. Research on block encoding of specific operators \cite{sturm2025efficient, kharazi2025explicit} seeks optimality by explicitly computing the operator's eigenvalues. Our criterion, in contrast, applies to an entire class, operates directly on the stencil coefficients, distinguishes the regimes where the construction is and is not optimal, and supplies the operator-norm floor that per-operator arguments leave implicit.

While the optimality criterion determines whether the shift-LCU construction incurs any avoidable subnormalization overhead, we introduce the moment structure to address a complementary question: how the encoded operator acts on spectrally restricted inputs and how its implementation cost changes near the zeros of the symbol.

Let $m$ denote the moment order of the stencil, namely, the lowest order at which its discrete moments do not vanish. We show that $m$ simultaneously determines \begin{enumerate} 
\item[(i)] the continuum differential operator the stencil approximates, an order-$m$ derivative; 
\item[(ii)] the vanishing structure of the stencil's Fourier symbol, which vanishes to order $m$ at the origin; and 
\item[(iii)] the block encoding cost: both the optimality verdict of the criterion above and the probability of successfully implementing the target operator, which scales as $\delta^{2m}$ in the spectral margin $\delta$ separating the input from the symbol's zeros.
\end{enumerate}
Leveraging information about the Fourier symbol and moment of an operator is common in classical PDE research (i, ii). However, to our knowledge, no prior work leverages the operator's moment when building block-encoding circuits.

In summary, we make the following contributions in this work:
\begin{itemize}

\item \textbf{Closed-form optimality certification for an operator family}
We derive a criterion that certifies whether the shift-LCU normalization of the block-encoding implementation attains the operator-norm floor, implying subnormalization optimality, for an entire translation-invariant family. This criterion applies uniformly in the grid size, with no per-operator eigenvalue computation, and with a quantified gap whenever it is not attained. (Proposition \ref{prop:opnorm})

\item \textbf{Moment order as the unifying parameter:} We show that a single integer, the moment order $m$ of a given stencil, determines the continuum operator it approximates (Section \ref{sec:moments}), the Fourier Symbol's vanishing structure (Section \ref{sec:spectral_analysis}), and, finally, the block encoding cost entering the
optimality verdict and the success-probability scaling (Section \ref{sec:block-encoding}, \ref{sec:safe-band-psuc}).

\item \textbf{A class-level success guarantee:} We prove a success-probability floor that holds uniformly over all inputs that retain a safe-band mass in the Fourier Spectrum. (Theorem \ref{thm:safe-band}, Section \ref{subsec:class-guarantee}).

\item \textbf{Recovery and extension of the optimal Laplacian encoding:} We recover the optimal Laplacian encoding of \cite{sturm2025efficient} as the gap-closed regime from our optimality argument and observe that our framework can use the same certificate to yield new gap-closed instances at higher even order, including the biharmonic stencil and the fourth-order Laplacian refinement (Section \ref{subsec:app-symmetric}).

\item \textbf{Explicit encodings for asymmetric operators:} We use our framework to obtain explicit constants for the block encoding of the advection-diffusion family, a multi-zero operator class for which no prior explicit spatial construction exists (Section \ref{subsec:app-advection}).

\end{itemize}

Section~\ref{sec:background} defines shift-stencil operators and cements their basic structural properties, including the one-to-one correspondence between coefficient functions and operators. Section~\ref{sec:moments} lays the foundation for spatial moments by defining the discrete moment structure and uses a Taylor Series expansion of the operator to derive the continuum differential operator being approximated. Section~\ref{sec:spectral_analysis} develops the dual spectral aspect, shows that the stencils are diagonalized in the quantum Fourier transform (QFT) basis, and analyzes the moment expansion of the Fourier symbol. Section~\ref{sec:block-encoding} begins by showing how to construct explicit block-encoding circuits for arbitrary shift-stencil operators, derives a closed-form expression for the probability of success by leveraging the symbol representation in the Fourier basis, and closes with the operator-norm optimality criterion with the subnormalization gap. Section~\ref{sec:safe-band-psuc} turns the exact expression for the probability of success into input-uniform guarantees, and  Section~\ref{sec:applications} applies the framework to families of operators. Finally, Section~\ref{sec:discussion} concludes the paper with a discussion.

\section{Background}\label{sec:background}
In this section, we introduce the key mathematical concepts used throughout this work. We begin by defining the grid, functions, and operators.
Let $\T^d$ be a flat torus 
\[
    \T^d := \R^d/\Z^d,
\]
identified with $[0,1]^d$ and periodic boundary conditions. This implies that for $x \in \T$, the function $v:\T^d\to \C$ defined on the torus $\T$ has the following periodic property: 
\[
    v(x^{(1)}, ..., 0, ..., x^{(d)})= v(x^{(1)}, ..., 1, ..., x^{(d)}).
\]
For $n \in \N$, we define the following parameters
\[
    N=2^n,\qquad h=\frac{1}{N},\qquad G_N:=\Z_N^d = \{0, 1, ..., N-1\}^d,
\]
to discretize this periodic function $v$ by sampling it to produce $v_h$ such that
\[
    v_h(j):=v(hj),\qquad j\in G_N.
\]
Consequently, this can be represented on an $n-qubit$ amplitude-encoded quantum state
\[
    \ket{v_h}:=
    \frac{1}{\norm{v_h}_2}
    \sum_{j\in G_N} v_h(j)\ket{j},
\]
for $v_h\neq 0$. For a finite set of $d$-dimensional offsets, $T \subset \Z^d$, we define the shift-stencil operator
\begin{equation}
    A_h := A_{N,T,c}:=\sum_{t\in T} c_t S_t .
    \label{eq:stencil-operator}
\end{equation}
where the offset coefficients $c_t\in \C$, $t\in T$, and $S_t$ is the quantum periodic shift operator with a wrap around of $N$,
\[
    S_t\ket{j} = \ket{j+t\bmod N}, ~j\in G_N.
\]

A square matrix is said to be circulant if each row is identical to the previous row, but shifted one to the right (or left). The top row is identical to the bottom row shifted to the right (or left). This property allows us to state that the stencils are diagonalizable in the Fourier domain. 

\begin{proposition}[Structural properties of shift-stencil operators]
\label{prop:shift-lcu}

For this shift-stencil operator  $A_h$ on $\T^d_N$, we make the following observations, some of which have been proven in other papers.
\begin{enumerate}[label=(\roman*),leftmargin=*]

    \item $A_h$ is already in linear-combination-of-unitaries \cite{childs2012hamiltonian} form, with
    $\ell^1$ subnormalization $\lambda := \sum_{t \in T} |c_t|$ and ancilla cost
    $m_a := \lceil \log_2 |T| \rceil$. 
        
    \item \cite{torosov2009} Since $A_h$ is a circulant operator (linear operator acting on periodic functions), it is diagonalized by the Quantum Fourier Transform (QFT) basis state
    \[
        A_h = F^{-1}\, \widehat{H}\, F,
        \qquad
        \widehat{H}(k) = \sum_{t \in T} c_t\, e^{-2\pi \ii \, k \cdot t / N},
        \quad k \in G_N
    \]
    where the function $\widehat H$ is called the Fourier symbol of the stencil.

    \item  Identifying a stencil with its coefficient function
    $c:\mathbb{Z}^d_N\to\mathbb{R}$ with offsets reduced modulo $N$ and support $\operatorname{supp} c := \{t\in\mathbb{Z}^d_N : c_t\neq 0\}$. The assignment $c\mapsto A_h=\sum_t c_t S_t$ puts stencils in one-to-one correspondence with the shift-stencil operators of the form in Eq. \eqref{eq:stencil-operator}. Distinct coefficient functions yield distinct operators, since the QFT jointly diagonalizes the shifts with distinct symbols.
    
\end{enumerate}
\end{proposition}
Properties (i) and (ii) give two dual representations of $A_h$: a spatial one in terms of coefficients $\{c_t\}$, and a spectral one in terms of the Fourier symbol $\widehat{H}$. The moment theory developed in this work is the bridge between them. Item (iii) is the fact that allows us to prove that we can use the SHIFT-LCU \eqref{eq:stencil-operator} form to block encode any translation-invariant, finite-supported discrete operator.

In the next section, we develop a moment-structured analysis that applies uniformly to all translation-invariant operators, independent of the application. The discrete moments provide the organizing structure for the entire construction, as they determine the operator family and enable explicit block-encoding circuits. The associated Fourier symbol is used to analyze postselection success probabilities and certify subnormalization optimality uniformly across those families.

\section{Moments and Spatial Analysis of the Stencil Operator}\label{sec:moments}
Moments are central to developing a unifying method for explicit circuit design across a wide range of translation-invariant operators. We start by defining discrete moments for a function on a $d$-dimensional torus $\T$.
\begin{definition}[Discrete moments and moment order]
For $\alpha := (\alpha_1, \alpha_2,...,\alpha_d)~\in \N^d$, define
\[
    |\alpha|:=\alpha_1+\cdots+\alpha_d,\qquad
    \alpha!:=\alpha_1!\cdots \alpha_d!,\qquad
    t^\alpha:=t_1^{\alpha_1}\cdots t_d^{\alpha_d},
\]
where $|\alpha|$ denotes the total degree (order) of $\alpha$. Then the $\alpha^{th}$ discrete moment for the stencil - $(T,~c)$ in Eq. \eqref{eq:stencil-operator}, is defined as
\begin{equation}
    \mu_\alpha := \sum_{t\in T} c_t t^\alpha .
    \label{eq:moment}
\end{equation}
Furthermore, the stencil is said to have moment order $m\ge0$ if every moment of order below $m$ vanishes, while at least one moment of order exactly $m$ does not:
\[
    \mu_\alpha=0, \qquad \forall~ |\alpha|<m
\]
and there exists at least one $\alpha$ with $|\alpha| = m$ such that $|\mu_\alpha|\neq0$.   
An equivalent form, which proves to be convenient when applying this definition, is that the moment order is the smallest degree $m$ at which the following homogeneous polynomial
\begin{equation}
    P_m(\xi) := \sum_{|\alpha|=m} \frac{\mu_\alpha}{\alpha!} \xi^\alpha,
    \label{eq:moment-order-poly}
\end{equation}
is not zero.
\label{def:moments}
\end{definition}
The following example solidifies the definition.
\begin{example}[$5$ point Laplacian]
We start from $d=2$, which mean that the $\alpha$ vector has the representation $\alpha=(\alpha_1,\alpha_2)$ with order $|\alpha|=\alpha_1+\alpha_2$. The finite difference discretization of the five-point Laplacian has the coefficients $c_{(0,0)}=-4$ and $c_{(\pm1,0)}=c_{(0,\pm1)}=1$. Every moment, $\mu_\alpha$, of order $\le 1$ vanishes,
\[
    \mu_{(0,0)}=-4+4=0,\qquad
    \mu_{(1,0)}=\mu_{(0,1)}=0,
\]
while at order $2$ three multi-indices have the form:
\[
    \mu_{(2,0)}=2,\qquad \mu_{(1,1)}=0,\qquad \mu_{(0,2)}=2 .
\]
Since all order-$\le1$ moments vanish, but some order-$2$ moment does not, the stencil has moment order $m=2$. The homogeneous polynomial from the alternative definition is therefore 
\[
    P_2(\xi)=\frac{\mu_{(2,0)}}{2!\,0!}\xi_1^2
            +\frac{\mu_{(1,1)}}{1!\,1!}\xi_1\xi_2
            +\frac{\mu_{(0,2)}}{0!\,2!}\xi_2^2
            =\xi_1^2+\xi_2^2=|\xi|^2 .
\]

\end{example}
Next, we define the norms used throughout this work.
\begin{definition}\label{def:2norm}
    The weighted Euclidean norm of $x$, written as $||x||_2$, is 
    \begin{equation*}
        ||x||_2 = \sqrt{ \sum_i |x_i|^2}.
    \end{equation*}
\end{definition}
We will also utilize the $L^2$ norm of a continuous function.
\begin{definition}\label{def:L2norm}
    The $L^2$ norm of a continuous function $f(x)$, written as $||x||_{L^2}$, is 
    \begin{equation*}
        ||f(x)||_{L^2} = \left( \int_{\mathbb{T}^d} |f(x)|^2 dx\right)^{\frac{1}{2}}.
    \end{equation*}
\end{definition}

Since we are sampling continuous functions on a discrete grid, it is important to understand the connection between the continuum and discrete operators. The following lemma does exactly that, showing that the scaled grid norm \(\ell_h^2\) is the natural Riemann-sum approximation of the continuum \(L^2\) norm. It further helps convert the uniform grid errors directly into norm errors.

\begin{lemma}[Scaled sampling norm and \(L^\infty\) to \(\ell_h^2\) control]
As previously, let \(G_N=\mathbb Z_N^d\) and \(h=1/N\). If we define the norm $||.||^2_{l^2_h}$ on a grid function \(x:G_N\to\C\) to be
\[
    \|x\|_{\ell_h^2}^2
    :=
    h^d\sum_{j\in G_N}|x(j)|^2,
\]
and the the norm $||.||_{l^\infty(G_N)}$ on a grid function \(x:G_N\to\C\) to be 
\[
    \|x\|_{l^\infty(G_N)}
    := \max_{j \in G_N} |x(j)|,
\] 
then the following holds: For every grid function \(x:G_N\to \C\),

    \[
        \|x\|_{\ell_h^2}
        \leq
        \|x\|_{\ell^\infty(G_N)}.
    \]
    Equivalently, if \(|x(j)|\leq \varepsilon_h\) for all \(j\in G_N\), then
    \[
        \|x\|_{\ell_h^2}\leq \varepsilon_h.
    \]

\label{lemma:discrete-continuous-norm}
\end{lemma}

To prove this, first note that
\[
    \|x\|_{\ell_h^2}^2
    =
    h^d\sum_{j\in G_N}|x(j)|^2
    \leq
    h^d\sum_{j\in G_N}\|x\|_{\ell^\infty(G_N)}^2.
\]

Since \(|G_N|=N^d\) and \(h=1/N\), we have
\[
    h^d|G_N|
    =
    h^dN^d
    =
    1.
\]
Therefore
\[
    \|x\|_{\ell_h^2}^2
    \leq
    \|x\|_{\ell^\infty(G_N)}^2,
\]
which gives
\[
    \|x\|_{\ell_h^2}
    \leq
    \|x\|_{\ell^\infty(G_N)}.
\]

Now, with the norm conventions established by Lemma \ref{lemma:discrete-continuous-norm}, we propose the following Theorem, which shows that the action of the shift-stencil operator on a sampled smooth function admits a Taylor expansion whose coefficients are precisely the discrete moments of the stencil.

\begin{theorem}[Taylor-moment expansion in physical space] Let $v$ be a periodic $r+1$ differentiable function on $\T^d$, ($v\in C_{per}^{r+1}(\T^d)$). Then, for $r>0$ and uniformly over $j\in G_N$, meaning that the remainder bound is independent of the grid point,  
we can expand the action of the operator $A_h$ (Eq.~\eqref{eq:stencil-operator}) on $v$ as
\begin{equation}
        (A_h v_h)(j)
    \;=\:
    \sum_{q=0}^{r} h^q
    \sum_{|\alpha|=q}
    \frac{(-1)^{|\alpha|}\mu_\alpha}{\alpha!}\,
    \partial^\alpha v(hj)
    \;+\;
    \mathcal{R}_{r+1,h}(j),
    \label{eq:taylor-moment-physical}
\end{equation}
where $\partial^\alpha v$ denotes the mixed partial derivative of total order $|\alpha|$, defined by
\[
    \partial^\alpha v
    \;:=\;
    \frac{\partial^{|\alpha|} v}
         {\partial x_1^{\alpha_1} \cdots \partial x_d^{\alpha_d}},
\]
and the remainder $\mathcal{R}_{r+1, h}(j)$ satisfies the scale-grid bound
\begin{equation}
    \norm{\mathcal{R}_{r+1,h}}_{\ell_{h}^2}
    \;\leq\;
    h^{r+1}\,
    \left(\sum_{t\in T}\abs{c_t}\,
    \frac{\norm{t}_1^{r+1}}{(r+1)!}\right)
    \max_{|\beta|=r+1}\norm{\partial^\beta v}_{L^\infty(\T^d)}.
    \label{eq:taylor-reminder}
\end{equation}

\label{thm:taylor-reminder-physical}    
\end{theorem}

\begin{proof}

First, we start with pointwise evaluation of $A_hv_h$ using periodicity. Fixing $j\in G_N$, and using the definition of $A_h$ and $S_t$,
\[
(A_h v_h)(j) = \sum_{t\in T} c_t\,(S_t v_h)(j)
            = \sum_{t\in T} c_t\, v_h\bigl((j-t)\bmod N\bigr)
\]
Since $v_h(k)=v(hk)$ for $k\in G_N$, and since
$h\big((j-t)\bmod N\big)\equiv h(j-t)\pmod{1}$ while $v$ is $1$-periodic in each coordinate, we have 
\[
    v_h\!\big((j-t)\bmod N\big)
    \;=\;
    v\!\big(h(j-t)\bmod 1\big)
    \;=\;
    v(hj-ht),
\]
where $(hj-ht)$ is a point on $\T^d$. Therefore
\begin{equation}
    (A_h v_h)(j)
    \;=\;
    \sum_{t\in T} c_t\, v(hj-ht).
    \label{eq:pointwise}
\end{equation}
Next, expanding $v(hj-ht)$ using the multivariate Taylor series applied at point $hj$ in the direction $ht$ results in
\begin{equation}
    v(hj-ht)
    \;=\;
    \sum_{|\alpha|\leq r}
    \frac{(-ht)^\alpha}{\alpha!}\,
    \partial^\alpha v(hj)
    \;+\;
    R_{r+1}(hj,t,h),
    \label{eq:multivariate-taylor}
\end{equation}
where
\begin{equation}
    R_{r+1}(hj,t,h)
    \;=\;
    (r+1)\sum_{|\alpha|=r+1}
    \frac{(-ht)^\alpha}{\alpha!}
    \int_0^1 (1-s)^{r}\,
    \partial^\alpha v(hj-hst)\,\dd s.
    \label{eq:integral-remainder}
\end{equation}
Letting $q=|\alpha|$ and using $(-ht)^\alpha = h^{|\alpha|}(-t)^\alpha$ yields
\[
    v(hj-ht)
    \;=\;
    \sum_{q=0}^{r}
    h^q\sum_{|\alpha|=q}
    \frac{(-t)^\alpha}{\alpha!}\,
    \partial^\alpha v(hj)
    \;+\;
    R_{r+1}(hj,t,h).
\]
We apply the operator $A_h$ by substituting this expanded form in Eq.~\eqref{eq:pointwise} to get
\begin{align*}
    (A_h v_h)(j)
    &\;=\;
    \sum_{q=0}^{r} h^q
    \sum_{|\alpha|=q}
    \frac{(-1)^{|\alpha|} \partial^\alpha v(hj)}{\alpha!}
    \underbrace{\sum_{t\in T} c_t\, t^\alpha}_{=\;\mu_\alpha}
    \;+\;
    \underbrace{\sum_{t\in T} c_t\, R_{r+1}(hj,t,h)}_{=:\;\mathcal{R}_{r+1,h}(j)}
    \\
    &\;=\;
    \sum_{q=0}^{r} h^q
    \sum_{|\alpha|=q}
    \frac{(-1)^{|\alpha|} \mu_\alpha}{\alpha!}\,
    \partial^\alpha v(hj)
    \;+\;
    \mathcal{R}_{r+1,h}(j),
\end{align*}
where we invoke the Eq.~\eqref{eq:moment} to arrive exactly at Eq.~\eqref{eq:taylor-moment-physical}. This concludes the first part of the proof. 

\noindent Finally, we prove the remainder bound starting from the integral in Eq.~\eqref{eq:integral-remainder}. The integral has an upper bound

\begin{align*}
    \abs{\int_0^1 (1-s)^{r}\, \partial^\alpha v(hj-hst)\,\dd s} 
    &\;\leq\; 
    \int_0^1 (1-s)^{r}\, \abs{\partial^\alpha v(hj-hst)\,\dd s}
    \\
    &\;\leq\; \max_{|\beta|=r+1}\norm{\partial^\beta v}_{L^\infty(\T^d)} 
    \int_0^1 (1-s)^{r}\,\dd s
    \\
    &\;=\; \max_{|\beta|=r+1}\norm{\partial^\beta v}_{L^\infty(\T^d)} \; \frac{1}{r+1}
    \label{eq:bounded-integral}
\end{align*}
where the first inequality is a triangle inequality, the second inequality is from pulling out the uniform pointwise bound, and the final bound is derived by evaluating the simple integral. Substituting this into Eq.~\eqref{eq:integral-remainder} observing $(ht)^\alpha= h^{r+1}t^\alpha$, $|t^{-\alpha}|=|t^\alpha|$ and using the multinomial identity,
\[
\sum_{|\alpha|=r+1} \frac{\abs{t^\alpha}}{\alpha!} \;=\; \frac{\norm{t}_1^{r+1}}{(r+1)!},
\]
we arrive at 
\begin{equation}
    \abs{R_{r+1}(hj,t,h)}
    \;\leq\;
    h^{r+1}\,
    \frac{\norm{t}_1^{r+1}}{(r+1)!}\,
    \max_{|\beta|=r+1}\norm{\partial^\beta v}_{L^\infty(\T^d)}.
    \label{eq:per-t-bound}
\end{equation}
Applying the triangle inequality in
$\mathcal{R}_{r+1,h}(j) = \sum_t c_t\, R_{r+1}(hj,t,h)$ and using
Eq.~\eqref{eq:per-t-bound} yields, uniformly in $j\in G_N$,
\[
    \abs{\mathcal{R}_{r+1,h}(j)}
    \;\leq\;
    h^{r+1}\,
    \left(\sum_{t\in T}\abs{c_t}\,
    \frac{\norm{t}_1^{r+1}}{(r+1)!}\right) \;
    \max_{|\beta|=r+1}\norm{\partial^\beta v}_{L^\infty(\T^d)}.
\]
Taking the supremum over $j$ keeps the upper bound unchanged, and using Lemma~\ref{lemma:discrete-continuous-norm} produces Eq.~\eqref{eq:taylor-reminder}, which concludes the proof.
\end{proof}

\begin{corollary}[Recovery of a differential operator at moment order $m$]
\label{cor:recovery}
Suppose $A_h$ has moment order $m\geq 1$ and let $v\in C_{\rm per}^{m+1}(\T^d)$. We define
\begin{equation}
    D_m v(x)
    \;:=\;
    \sum_{|\alpha|=m}
    \frac{\mu_\alpha}{\alpha!}\,
    (-\partial)^\alpha v(x)
    \;=\;
    (-1)^m \sum_{|\alpha|=m} \frac{\mu_\alpha}{\alpha!}\, \partial^\alpha ,
    \qquad x\in\T^d,
    \label{eq:Dm}
\end{equation}
the moment-weighted linear combination of mixed partial derivatives of
total order $m$. Then, uniformly over $j\in G_N$,
\begin{equation}
    (A_h v_h)(j)
    \;=\;
    h^m\,(D_m v)(hj)
    \;+\;
    \mathcal{R}_{m+1,h}(j),
    \label{eq:recovery}
\end{equation}
with $\norm{\mathcal{R}_{m+1,h}}_{\ell_{h}^2} = O(h^{m+1})$
given explicitly by Eq.~\eqref{eq:taylor-reminder} at $r=m$.
\end{corollary}
\begin{proof}
We apply Theorem~\ref{thm:taylor-reminder-physical} with $r=m$. From the definition of moment order: $\mu_\alpha = 0$ for all $|\alpha|<m$, which implies that the terms $q = 0, 1, \ldots, m-1$ in Eq.~\eqref{eq:taylor-moment-physical} vanish. The leading surviving term at $q=m$ is $h^m (D_m v)(hj)$ by Eq.~\eqref{eq:Dm}, and the remainder bound is Eq.~\eqref{eq:taylor-reminder} recovered with the substitution $r=m$.
\end{proof}

Concretely, $h^{-m}A_h$ acts on smooth sampled functions as the differential operator $D_m$ up to a remainder of order $h$. We have shown that the moment data $(\mu_\alpha)_{\alpha=m}$ alone determines which continuum operator the stencil approximates.

\begin{corollary}[Sharper expansion when additional moments vanish]
Suppose the stencil has moment order $m$ and that for some $\rho\geq 1$,
\[
    \mu_\alpha=0
    \qquad
    \text{for all }m<|\alpha|<m+\rho.
\]
If $v\in C_{\rm per}^{m+\rho+1}$, then applying Theorem \ref{thm:taylor-reminder-physical} with $r=m$ and using the definition in Eq.~\eqref{eq:recovery},
\[
    A_hv_h
    \;=\;
    h^m(D_mv)_h + h^{m+\rho}(D_{m+\rho}v)_h + O_{\ell_h^2}(h^{m+\rho+1}).
\]
Here, we use the big-O notation to show the leading power of $h$ that bounds the remainder term.
\end{corollary}
This is particularly important since, for symmetric stencils such as the centered second difference, the third moment vanishes, so the next error term is $O(h^4)$ rather than $O(h^3)$. We move on to the next section, which develops the Fourier symbol analysis of the periodic operator on the discretized grid.

\section{Spectral Analysis of the Stencil Operator}\label{sec:spectral_analysis}
We start this section by settling the convention for frequencies used throughout this work. The frequencies live on a torus $\mathbb{T}^d_{2\pi} := \mathbb{R}^d/(2\pi\mathbb{Z})^d$ and the symbol $\widehat H(\theta) = \sum_{t\in T} c_t e^{-i\theta\cdot t}$ is therefore $2\pi$-periodic in each coordinate and real-analytic on $\mathbb{R}^d$. We write frequency points through representatives in the closed box $[-\pi,\pi]^d$, with opposite faces corresponding to the same frequency (i.e., $\theta$ and $\theta + 2\pi k$, $k\in\mathbb{Z^d}$). This means that $\theta = -\pi$ and $\theta = \pi$ are the same point. Furthermore, we note that all distances are torus distances, 
\[ 
\dist_{\mathbb{T}}(\theta,\zeta):= \min_{k\in\mathbb{Z}^d} \norm{\theta-\zeta+2\pi k}_2.
\]
Then, from the topology of a torus, $\mathbb{T}^d_{2\pi}$ is compact, and since $\widehat H$ is continuous, the suprema of $\lvert\widehat H\rvert$ over the zone are the maxima.

We further clarify that since the discrete grid $\Theta_N$ consists of the $N^d$ distinct torus points $\theta_k \equiv 2\pi k/N \pmod{2\pi}$, $k\in\mathbb{Z}^d_N$, each frequency class, including the zone-edge class $\theta\equiv\pi$, is counted exactly once in every sum over $\Theta_N$.

With the frequency convention in place, we now characterize the spectral structure of the shift-stencil operator $(A_h)$, which will enable us to prove the central results of this paper. First, we show that the lattice shift operators are diagonalized by the discrete Fourier basis, so the eigenvalues of $A_h$ coincide with the values of the Fourier symbol $\widehat H$. This observation significantly simplifies our analysis in Section \ref{sec:block-encoding}, which now requires us to consider $(\widehat H)$ rather than the spectrum and operator norm of $A_h$. We then analyze the local behavior $\widehat H$ using stencil moments. Together, these results provide the spectral foundation for Section \ref{sec:block-encoding}, where we construct the shift-LCU block encoding, derive its post-selection success probability, and finally derive the subnormalization optimality conditions.
 
\begin{definition}[Discrete Brillouin grid and Fourier basis states]
\label{def:brillouin}
The \emph{discrete Brillouin grid} of $G_N$ is
\[
    \Theta_N
    \;:=\;
    \left\{\theta_k = \frac{2\pi}{N}\,k \;:\; k\in\Z_N^d\right\}
    \;\subset\; torus~points~as~above.
\]
For each $\theta\in\Theta_N$, the associated \emph{Fourier basis state} 
on $n$-qubits is
\[
    \ket{\chi_\theta}
    \;:=\;
    N^{-d/2}\sum_{j\in G_N} e^{\ii\, j\cdot \theta}\ket{j}.
\]
The family $\{\ket{\chi_\theta}\}_{\theta\in\Theta_N}$ is an orthonormal
basis of the $n$-qubit Hilbert space, obtained from the computational
basis by the (inverse) Quantum Fourier Transform.
\end{definition}

\begin{theorem}[Fourier diagonalization of shift-stencils]
\label{thm:diagonalization}
For every $t\in\Z^d$ and $\theta\in\Theta_N$,
\begin{equation}
    S_t\ket{\chi_\theta}
    \;=\;
    e^{-\ii\, t\cdot \theta}\,\ket{\chi_\theta}.
    \label{eq:shift-eigen}
\end{equation}
Consequently, $A_h$ is diagonal in the Fourier basis:
\begin{equation}
    A_h\ket{\chi_\theta}
    \;=\;
    \widehat H(\theta)\,\ket{\chi_\theta},
    \qquad
    \widehat H(\theta)
    \;=\;
    \sum_{t\in T} c_t\, e^{-\ii\, t\cdot \theta}.
    \label{eq:diag}
\end{equation}
\end{theorem}

\begin{proof}
By the definition of $\ket{\chi_\theta}$ and the action of $S_t$,
\[
    S_t\ket{\chi_\theta}
    \;=\;
    N^{-d/2}\sum_{j\in G_N} e^{\ii\, j\cdot\theta}\,\ket{(j+t)\bmod N}.
\]
Reindexing the sum with $\ell := (j+t)\bmod N \in G_N$, and using that
$e^{\ii\,(\,\cdot\,)\cdot\theta}$ is $N$-periodic in each coordinate
since $\theta\in\Theta_N$,
\[
    S_t\ket{\chi_\theta}
    \;=\;
    N^{-d/2}\sum_{\ell\in G_N} e^{\ii\,(\ell - t)\cdot\theta}\ket{\ell}
    \;=\;
    e^{-\ii\, t\cdot\theta}\,
    \underbrace{N^{-d/2}\sum_{\ell\in G_N} e^{\ii\,\ell\cdot\theta}\ket{\ell}}_{=\;\ket{\chi_\theta}},
\]
which proves Eq.~\eqref{eq:shift-eigen}. Linearity of $A_h = \sum_t c_t S_t$
then gives Eq.~\eqref{eq:diag}.
\end{proof}

\subsection{Fourier symbol expansion}
We analyze the low-frequency behavior of $A_h$ by considering its Fourier symbol $\widehat H$ to be a real-analytic function on $\R^d$ (equivalently, on the torus, representatives in $[-\pi,\pi]^d$ due to the periodicity of the shift). We denote $\theta\in\R^d$ to be a continuous variable replacing $2\pi k/N$, which results in the following representation: 
\begin{equation}
    \widehat H(\theta)
    \;:=\;
    \sum_{t\in T} c_t\, e^{-\ii\, \theta\cdot t}.
    \label{eq:symbol}
\end{equation}
We now expand $\widehat H$ in a Taylor series around $\theta=0$ and identify the coefficients with the moments of the stencil.

\begin{theorem}[Moment expansion of the Fourier symbol]
\label{thm:symbol-expansion}
Let $r\geq 0$. Then for every $\theta\in\R^d$, 
\begin{equation}
    \widehat H(\theta)
    \;=\;
    \sum_{q=0}^{r}
    \sum_{|\alpha|=q}
    \frac{(-\ii)^{|\alpha|}\,\mu_\alpha}{\alpha!}\,\theta^\alpha
    \;+\;
    R_{r+1}(\theta),
    \label{eq:symbol-expansion}
\end{equation}
with remainder bound
\begin{equation}
    \abs{R_{r+1}(\theta)}
    \;\leq\;
    \frac{\norm{\theta}_2^{r+1}}{(r+1)!}\,
    \sum_{t\in T}\abs{c_t}\,\norm{t}_2^{r+1}.
    \label{eq:symbol-remainder}
\end{equation}
In particular, if $(T,c_t)$ has moment order $m$, then
\begin{equation}
    \widehat H(\theta)
    \;=\;
    (-\ii)^m \sum_{|\alpha|=m}
    \frac{\mu_\alpha}{\alpha!}\,\theta^\alpha
    \;+\;
    O\!\big(\norm{\theta}_2^{m+1}\big)
    \qquad\text{as }\theta\to 0.
    \label{eq:symbol-leading}
\end{equation}
\end{theorem}

\begin{proof}
For every $x\in\R$ and integer $r\geq 0$, the Taylor remainder for
$e^{-\ii x}$ satisfies
\begin{equation}
    \left|\, e^{-\ii x}
    - \sum_{q=0}^{r}\frac{(-\ii x)^q}{q!} \,\right|
    \;\leq\;
    \frac{|x|^{r+1}}{(r+1)!},
    \label{eq:scalar-exp-remainder}
\end{equation}
since $|(d^k/dx^k)\,e^{-\ii x}| = 1$ for every $k$. Applying
Eq.~\eqref{eq:scalar-exp-remainder} with $x = \theta\cdot t$,
\[
    e^{-\ii\,\theta\cdot t}
    \;=\;
    \sum_{q=0}^{r}\frac{(-\ii\,\theta\cdot t)^q}{q!}
    \;+\;
    r_{r+1}(\theta,t),
    \qquad
    |r_{r+1}(\theta,t)|
    \leq
    \frac{|\theta\cdot t|^{r+1}}{(r+1)!}.
\]
By the multinomial theorem,
$(\theta\cdot t)^q = \sum_{|\alpha|=q}\frac{q!}{\alpha!}\,t^\alpha\theta^\alpha$,
so
\[
    \frac{(-\ii\,\theta\cdot t)^q}{q!}
    \;=\;
    (-\ii)^q \sum_{|\alpha|=q}\frac{t^\alpha\theta^\alpha}{\alpha!}.
\]
Multiplying by $c_t$, summing over $t\in T$, and using
$\mu_\alpha = \sum_t c_t t^\alpha$:
\[
    \widehat H(\theta)
    \;=\;
    \sum_{q=0}^{r}(-\ii)^q\sum_{|\alpha|=q}
    \frac{\mu_\alpha\,\theta^\alpha}{\alpha!}
    \;+\;
    R_{r+1}(\theta),
    \qquad
    R_{r+1}(\theta) := \sum_{t\in T} c_t\, r_{r+1}(\theta,t),
\]
which gives Eq.~\eqref{eq:symbol-expansion} since $(-\ii)^q = (-\ii)^{|\alpha|}$
when $|\alpha|=q$.

For the remainder bound, the triangle inequality and Cauchy--Schwarz
give
\[
    \abs{R_{r+1}(\theta)}
    \;\leq\;
    \sum_{t\in T}\abs{c_t}\,
    \frac{|\theta\cdot t|^{r+1}}{(r+1)!}
    \;\leq\;
    \frac{\norm{\theta}_2^{r+1}}{(r+1)!}
    \sum_{t\in T}\abs{c_t}\,\norm{t}_2^{r+1},
\]
which is Eq.~\eqref{eq:symbol-remainder}. Finally, Eq.~\eqref{eq:symbol-leading}
follows by specializing to $r=m$, dropping all terms with $|\alpha|<m$
by the moment-order hypothesis, and absorbing the explicit bound into
$O(\|\theta\|_2^{m+1})$.
\end{proof}

The choice of expanding around $\theta=0$ is not arbitrary, but rather because it is significant for a wide range of target operators and their block encodings. First, for smooth periodic inputs, the symbol decays as $|\theta|^{-k}$, so its spectral mass is concentrated near the origin. Second, one of the primary interests in this paper is stencils with moment order $m \geq 1$, which approximate differential operators (Corollary~\ref{cor:recovery}); for these, $\widehat H(0) = \mu_{\mathbf 0} = 0$, so $\theta = 0$ lies in the zero set of the symbol. Third, the local behavior of $\widehat{H}$ near its zero set controls the post-selection success probability of the shift-LCU block encoding. At frequencies where $\widehat{H}$ is small, $A_h$ attenuates input amplitude, and the safe-band analysis of Section~\ref{sec:safe-band-psuc} quantifies this dependence explicitly.

\begin{remark}[Consistency with the recovered operator]
The leading term in Eq.~\eqref{eq:symbol-leading} is precisely $(-\ii)^m P_m(\theta)$, where $P_m$ is the moment polynomial from Eq.~\eqref{eq:moment-order-poly}. With the convention from Eq.~\eqref{eq:symbol}, the symbol of the differential operator $D_m$ from Corollary~\ref{cor:recovery} acts on the plane wave $e^{ \ii\theta\cdot x}$ as multiplication by $(-\ii)^m \sum_{|\alpha|=m}\frac{\mu_\alpha}{\alpha!}\theta^\alpha$. This is the spectral counterpart of Corollary~\ref{cor:recovery}: the same moment data $\{\mu_\alpha\}_{|\alpha|=m}$ determine both the physical-space leading behavior $h^m(D_m v)_h$ and the frequency-space leading behavior of the symbol $\widehat H(\theta)$ near $\theta=0$. 
\end{remark}

\section{Block encoding the Shift-LCU}
\label{sec:block-encoding}
In the preceding sections, we developed the spatial and spectral analysis of the operator $A_h$, showing that moments are central to isolating the continuum approximation in the spatial domain and the local structure of the Fourier symbol. Now, we focus on implementing the operator and observe that the same moments govern the cost of block encoding it on a quantum computer.

\subsection{Block encoding primitives for the shift stencil}
As per Proposition~\ref{prop:shift-lcu}(i), the structure of $A_h$ is already in the Linear Combination of Unitaries form. We adopt the standard notion of block-encoding proposed in \cite{gilyen2019quantum}.

\begin{definition}[Block encoding]
\label{def:block-encoding}
Let $A_h$ act on $n$ qubits. A unitary $U$ on $m_a + n$ qubits is an $(\lambda, m_a, \varepsilon)$ block encoding of $A$ if
\[
    \norm{\,A_h
    \;-\;
    \lambda\,
    \big(\bra{0}^{\otimes m_a}\otimes I_n\big)\,
    U\,
    \big(\ket{0}^{\otimes m_a}\otimes I_n\big)}_{\mathrm{op}}
    \;\leq\;
    \varepsilon.
\]
Here, the parameter $\lambda$ is the subnormalization, $m_a$ the number of ancilla qubits, $\varepsilon$ the encoding error and $\norm{.}_{\mathrm{op}}$ is the operator norm.
\end{definition}
The following proposition uses this definition to implement the encoding for our operator $A_h$. We use the fact that the operator is already expressed as a linear combination of shifts, $\{S_t\}_{t\in T}$, with $\ell^1$ weights $\lambda = \sum_{t \in T} |c_t|$.

\begin{proposition}[Block encoding implementation for shift-LCU]
We use the Prepare-Select method for the implementation. Let us represent $c_t$ in Euler's form $c_t = \abs{c_t}\,e^{\ii\phi_t}$ for $t\in T$, $\lambda = \sum_{t\in T}\abs{c_t}$ and $m_a = \lceil\log_2\abs{T}\rceil$. The right and left Prepare unitaries are defined by their action on the ancillary qubits:
\[
    \mathrm{PREP}_R\ket{0}^{\otimes m_a}
    \;=\;
    \frac{1}{\sqrt\lambda}\sum_{t\in T}\sqrt{\abs{c_t}}\,\ket{t},
    \qquad
    \mathrm{PREP}_L\ket{0}^{\otimes m_a}
    \;=\;
    \frac{1}{\sqrt\lambda}\sum_{t\in T}\sqrt{\abs{c_t}}\,e^{-\ii\phi_t}\ket{t},
\]
while the Select unitaries are defined as operators with ancillary control acting on the system qubits,
\[\mathrm{SEL} = \sum_{t\in T}\ket{t}\!\bra{t}\otimes S_t.\]
Now, a  $(\lambda, m_a, 0)$ block encoding of $A_h$ can be implemented by
\[
    U \;:=\;
    \big(\mathrm{PREP}_L^\dagger\otimes I_n\big)\,
    \mathrm{SEL}\,
    \big(\mathrm{PREP}_R\otimes I_n\big)
\]
\begin{proof}
Acting on $\ket{0}^{\otimes m_a}\ket{j}$ and applying the three factors
in turn,
\[
    \big(\mathrm{PREP}_R\otimes I_n\big)\ket{0}^{\otimes m_a}\ket{j}
    =
    \frac{1}{\sqrt\lambda}\sum_{t\in T}\sqrt{\abs{c_t}}\,\ket{t}\ket{j}
    \;\xmapsto{\;\mathrm{SEL}\;}\;
    \frac{1}{\sqrt\lambda}\sum_{t\in T}\sqrt{\abs{c_t}}\,\ket{t}\,S_t\ket{j}.
\]
Projecting the ancilla onto $\ket{0}^{\otimes m_a}$ after
$\mathrm{PREP}_L^\dagger$ and using
$\bra{0}^{\otimes m_a}\mathrm{PREP}_L^\dagger\ket{t}
= \lambda^{-1/2}\sqrt{\abs{c_t}}\,e^{\ii\phi_t}$,
\[
    \big(\bra{0}^{\otimes m_a}\otimes I_n\big)\,U\,
    \big(\ket{0}^{\otimes m_a}\otimes I_n\big)\ket{j}
    =
    \frac{1}{\lambda}\sum_{t\in T}\abs{c_t}\,e^{\ii\phi_t}\,S_t\ket{j}
    =
    \frac{1}{\lambda}\sum_{t\in T} c_t\,S_t\ket{j}
    =
    \frac{1}{\lambda}A_h\ket{j}.
\]
Since this holds for every basis state $\ket{j}$, the encoding is exact and thus
$\varepsilon = 0$.
\end{proof}
\label{prop:be-shiftlcu}
\end{proposition}
The two Prepare circuits differ in implementation by a factor of $e^{\ii\phi_t}$. For the operators of interest in this paper, the coefficients of the LCU are real, and therefore the parameter $\phi$ is restricted to the values $\{0, \pi\}$.  This implies that the Prepare circuits differ only by a $-$ sign, which can be implemented simply using a single-qubit $Z$ gate on the branch-indicator qubit register.

\subsection{Success Probability}
The probability of success is a central measure for block encoding methods. We derive a closed-form success probability by exploiting the fact that the operator is diagonal in the Fourier basis (Theorem \ref{thm:diagonalization}). 

\begin{theorem}[Closed-form success probability]
\label{thm:success-closedform}
Let $U$ be the block encoding as in Proposition \ref{prop:be-shiftlcu}, and let $\ket{v}$ be a normalized $n$-qubit state with Fourier expansion
\[
    \ket{v}
    =
    \sum_{\theta\in\Theta_N}\tilde v(\theta)\,\ket{\chi_\theta},
    \qquad
    \sum_{\theta\in\Theta_N}\abs{\tilde v(\theta)}^2 = 1.
\]
Applying $U$ to $\ket{0}^{\otimes m_a}\ket{v}$ and measuring the ancilla
register, the probability of the outcome $\ket{0}^{\otimes m_a}$ is
\begin{equation}
    p_{\mathrm{succ}}(\ket{v})
    \;=\;
    \frac{1}{\lambda^2}\,\norm{A_h\ket{v}}_2^2
    \;=\;
    \frac{1}{\lambda^2}
    \sum_{\theta\in\Theta_N}
    \abs{\widehat H(\theta)}^2\,\abs{\tilde v(\theta)}^2.
    \label{eq:succ-closedform}
\end{equation}
If post-selection results in success, the circuit simulates the action of operator $A_h$ on the system register, resulting in $A_h\ket{v}/\norm{A_h\ket{v}}_2$.
\end{theorem}
\begin{proof}

From the definition of Block encoding \ref{def:block-encoding} and Proposition \ref{prop:be-shiftlcu}, 
\[U\ket{0}^{\otimes m_a}\ket{v} = \ket{0}^{\otimes m_a}\,\lambda^{-1}A_h\ket{v} + \ket{\perp}\]
where $\ket{\perp}$ is the state orthogonal to the desired post-selection state $\ket{0}^{\otimes m_a}$. Using Borne's rule, the probability of selecting the  $\ket{0}^{\otimes m_a}$ state is  $\norm{\lambda^{-1}A_h\ket{v}}_2^2$ and the post-measurement system state is the normalization of $A_h\ket{v}$. Representing $A_h$ and $\ket{v}$ in the $\{\ket{\chi_\theta}\}$ basis and applying Theorem \ref{thm:diagonalization} to diagonalize $A_h\ket{\chi_\theta} = \widehat H(\theta)\ket{\chi_\theta}$:
\[
    A_h\ket{v}
    \;=\;
    \sum_{\theta\in\Theta_N}\tilde v(\theta)\,\widehat H(\theta)\,\ket{\chi_\theta}.
\]
Use the fact that $\{\ket{\chi_\theta}\}$-basis is orthonormal to get $\norm{A_h\ket{v}}_2^2 = \sum_{\theta}\abs{\widehat H(\theta)}^2\abs{\tilde v(\theta)}^2$, which gives Eq.~\eqref{eq:succ-closedform}. 
\end{proof}

\subsection{The operator norm as an optimality benchmark}

The cost of block encoding is controlled by the subnormalization $\lambda =  \sum_{t\in T}\abs{c_t}$, and the probability of success scaling is quantified in Theorem \ref{thm:success-closedform}. Since a lower $\lambda$ value directly improves success probability at the rate of $\lambda^{-2}$, it is important to quantify its lower bound, which would therefore be deemed as an optimal subnormalization. We notice that, given a $\lambda'$ there is no $(\lambda',m_a,0)$ block encoding of $A_h$ which can have $\lambda' < \norm{A_h}_{\mathrm{op}}$. Therefore, the operator norm on $A_h$ is the floor to the achievable subnormalization, and we say that the LCU design is optimal when our $\lambda$ attains it. The following proposition solidifies this statement.

\begin{proposition}[Operator norm and the subnormalization gap]
\label{prop:opnorm}
The shift-stencil operator satisfies
\begin{equation}
    \norm{A_h}_{\mathrm{op}}
    \;=\;
    \max_{\theta\in\Theta_N}\abs{\widehat H(\theta)}
    \;\leq\;
    \norm{\widehat H}_\infty
    \;\leq\;
    \lambda,
    \label{eq:opnorm-chain}
\end{equation}
where $\norm{\widehat H}_\infty := \max{\theta\in[-\pi,\pi]^d}\abs{\widehat H(\theta)}$.
The last inequality is an equality if the terms $\{c_t\,e^{ -\ii\,\theta^\star\cdot t}\}_{t\in T}$ share a common phase at some $\theta^\star$, e.g., $\norm{\widehat H}_\infty=\lambda$ whenever all $c_t\geq0$.
\end{proposition}

\begin{proof}
By Theorem \ref{thm:diagonalization}, $A_h = F^{-1}\,\widehat H\,F$ where $F$ is unitary and $\widehat H$ is diagonal with entries $\{\widehat H(\theta)\}_{\theta\in\Theta_N}$. The singular values evaluated from the right-hand side of $A_h$ are therefore $\{|\widehat H(\theta)|\}_{\theta\in\Theta_N}$. The first equality follows from the definition of the operator norm as its maximum. The discrete maximum is bounded by the continuous supremum, giving the second relation. Finally, the third follows by applying the triangle inequality, which gives, for every
$\theta$,
\[
    \abs{\widehat H(\theta)}
    =
    \Big|\sum_{t\in T} c_t\,e^{-\ii\,\theta\cdot t}\Big|
    \;\leq\;
    \sum_{t\in T}\abs{c_t}
    =
    \lambda,
\]
with equality precisely when all summed terms have a common phase in $\C$. Trivially, if $c_t\geq 0$ for all $t$, this holds at $\theta^\star = 0$.
\end{proof}
\begin{remark} 
This means that for any $(\lambda', m_a, 0)$ block encoding, $A_h/\lambda'$ is a corner block of a unitary, hence $\|A_h/\lambda'\|_{op} \le 1$ and $\lambda' \ge \|A_h\|_{op}$. The operator norm is therefore a floor over all exact encodings, and the common-phase condition certifies when the shift-LCU attains it. 
\end{remark}

When $\lambda = \lambda_{opt} = \norm{\widehat H}_\infty$, the LCU construction is optimal in terms of subnormalization, and no alternative construction can improve it. For the cases where $\lambda > \norm{\widehat H}_\infty$, which is strictly greater than the lower bound, the symbol exhibits cancellation that the $\ell^1$ weight does not see, and an operator-specific encoding may close the gap. For the symmetric Laplacian stencil, our framework produces the optimal subnormalization, whereas for the difference-of-Gaussian kernel, the $ \|\widehat H\|_\infty<\lambda=2< \lambda_{opt}$ bound \cite{mahmud2026explicit}, although it facilitates better circuit construction, decreases the probability of success.

In this section, we derived the closed-form expression and discussed the optimal subnormalization for the block-encoding success probability. The closed-form expression, Eq.~\eqref{eq:succ-closedform}, is exact but does not by itself yield a useful input-independent guarantee. However, it gives the following insight: For a moment order (Eq.~\eqref{def:moments}) $m\geq1$, $\widehat H(0)=\mu_0=0$ by Theorem \ref{thm:symbol-expansion}, the symbol vanishes in the set
\[
    Z \;:=\; \{\theta \in [-\pi, \pi]^d: \widehat H (\theta) \;=\; 0 \} \;\ni\; 0,
\]
which means that $\widehat H$ is minimum whenever the input carries a spectral mass on or near the zero set $Z$. A guarantee for the success probability will therefore require restricting ourselves to inputs that are bounded away from $Z$. In Section \ref{sec:safe-band-psuc}, we quantify this observation to provide a bound for the probability of success in terms of the moment order $m$, the distance $\delta$ to the zero set, and spectral mass $\eta$ retained in the defined ``safe-band".

\section{Safe-Band Success-Probability Bounds}\label{sec:safe-band-psuc}

As observed from the closed-form expression in Eq.~\eqref{eq:succ-closedform}, the symbol vanishes in the zero set regardless of the input. The route to deriving input-uniform bounds on the success probability involves restricting our attention to inputs whose spectral mass lies in the regions away from the symbol's zero set. A central observation is that the differential operators act trivially on near-constant data. In this Section, we use this observation and the insights from the previous section to provide a concrete claim: for an elliptic operator of moment order $m$, every input retaining spectral mass $\eta_\delta$ outside a $\delta$ neighborhood of $Z$ succeeds with probability at least
\[
    p_{\mathrm{succ}}(\ket{v})
    \;\geq\;
    \frac{\kappa_m^2\,\delta^{2m}}{4\,\lambda^2}\,\eta_\delta(\ket{v}),
\]
where $\kappa_m$ is the moment constant, the spectral mass is $\eta_\delta$, and $\lambda$ is the block encoding subnormalization. We start by defining a safe-band around the zero set, bound $\widehat H$ from below on it as a function of $m$, and use it to arrive at the $P_{succ}$ guarantee.

\begin{definition}[Zero set, safe-band, and spectral mass]
\label{def:safe-band}
For a shift-stencil operator $A_h$ with symbol $\widehat H$, its
\emph{zero set} is
\[
    Z
    \;:=\;
    \{\theta\in[-\pi,\pi]^d : \widehat H(\theta) = 0\}.
\]

For $\delta > 0$ the \emph{safe-band} is
\[
    \Omega_\delta
    \;:=\;
    \big\{\theta\in[-\pi,\pi]^d :
    \dist_2(\theta, Z)\geq \delta\big\},
    \qquad
    \dist_2(\theta,Z) := \min_{\zeta\in Z}\dist_{\mathbb T}(\theta,\zeta),
\]
and the \emph{safe-band spectral mass} of a normalized state
$\ket{v} = \sum_\theta \tilde v(\theta)\ket{\chi_\theta}$ is
\[
    \eta_\delta(\ket{v})
    \;:=\;
    \sum_{\theta\in\Theta_N\cap\,\Omega_\delta}
    \abs{\tilde v(\theta)}^2
    \;\in\;[0,1].
\]
\end{definition}

\subsection{Lower bound of the symbol}\label{subsec:lowerbound}
To develop a guarantee for success probability, we must first find a lower bound for $\abs{\widehat H}$ inside the defined safe-band. The Taylor series describes the symbol near $\theta=0$ where the behavior, as per Theorem \ref{thm:symbol-expansion}, is $\widehat H(\theta)\approx(-\ii)^m P_m(\theta)$ where $P_m$ is the homogeneous polynomial Eq.~\eqref{eq:moment-order-poly}. The symbol vanishes to order $m$, and the manner of that vanishing is encoded in the shape of $P_m$. This allows us to develop a lower bound for the $P_m$ and consequently for $\abs{\widehat H}$.

We analyze the shape of the symbol by using a vector representation of $\theta$: $\theta=\norm{\theta}_2 \xi$, which is the normalized unit vector, $\xi$, multiplied by the vector norm. Since, according to our definition, Eq.~\eqref{eq:moment-order-poly}, $P_m$ is a homogeneous polynomial of order $m$, $\abs{P_m(\theta)}=\norm{\theta}_2^m\abs{P_m(\xi)}$. This representation allows analysis of the shape of $P_m$ on the unit sphere described by $\xi$.

If the operator is elliptic, the rate at which $P_m$ climbs out of $0$ will be non-zero in all directions. We define a minimum rate as
\begin{equation}
\kappa_m
\;:=\;
\min_{\norm{\xi}_2 = 1}\abs{P_m(\xi)}
\;>\;0,
\label{eq:ellipticity}
\end{equation}
which directly asserts a bound $|P_m(\theta)| \ge \kappa_m \|\theta\|_2^m$ by homogeneity, inheriting this bound near the origin (as we show in the next lemma, \ref{lem:symbol-lower}). We observe that several useful operators are not only elliptic but also isotropic, meaning that the rate $\kappa_m$ is the same in all directions. The Laplacian case is an isotropic example in which the rate is constant, $\kappa_2=\tfrac14$, yielding $P_2(\xi)=\tfrac14\norm{\xi}_2^2$, which consequently produces the leading-order value of the symbol. We defer the full derivation of the quantities to Section \ref{sec:applications}.

To develop a universal lower bound for the symbol, we divide the Brillouin zone into two parts: (i) a neighborhood of the origin, where the leading-order moment analysis applies directly, together with homogeneity and ellipticity, and (ii) the region away from the origin, where the remainder term may become significant, and compactness is used. This leads to the following lemma, which establishes a lower bound for the symbol. For simplicity, we first assume that the origin is the only zero of the symbol, and later discuss how this assumption can be relaxed.

\begin{lemma}[Safe-band lower bound on the symbol]
\label{lem:symbol-lower}
Let $A_h$ have moment order $m\geq 1$, let the origin be the only zero of the symbol, $Z=\{0\}$, and let $P_m$ be elliptic (satisfying Eq.~\eqref{eq:ellipticity}). Define the remainder term of the symbol from Eq.~\eqref{eq:symbol-remainder} bound as $B:= \tfrac{1}{(m+1)!}\sum_{t\in T}\abs{c_t}\norm{t}_2^{m+1}$ and 
\begin{equation}
    \rho
    \;:=\;
    \min\!\left\{\pi,\;\frac{\kappa_m}{2B}\right\},
    \qquad
    c_\star
    \;:=\;
    \min_{\norm{\theta}_2\geq\rho}\abs{\widehat H(\theta)} \;>\;0 .
    \label{eq:rho-cstar}
\end{equation}
Then for every $\theta\in[-\pi,\pi]^d$, 
\begin{equation}
    \abs{\widehat H(\theta)}
    \;\geq\;
    \begin{cases}
        \dfrac{\kappa_m}{2}\,\norm{\theta}_2^{\,m},
        & \norm{\theta}_2 < \rho,\\[2ex]
        c_\star,
        & \norm{\theta}_2 \geq \rho ,
    \end{cases}
    \label{eq:symbol-lower-piecewise}
\end{equation}
and consequently, for every $\delta>0$, with $g_m(\delta):=\tfrac12\min\{\kappa_m\,\delta^m,\,2c_\star\}$,
\begin{equation}
    \abs{\widehat H(\theta)}
    \;\geq\;
    g_m(\delta)
    \qquad\text{for all }\theta\in\Omega_\delta \;\ \text{as defined in Def.~\ref{def:safe-band}.}
    \label{eq:symbol-lower-uniform}
\end{equation}
\end{lemma}

\begin{proof}
\emph{Near the origin.}
By Theorem \ref{thm:symbol-expansion} with $\mu_\alpha=0$ for
$\abs{\alpha}<m$,
\[
    \widehat H(\theta)
    =
    (-\ii)^m P_m(\theta) + R_{m+1}(\theta),
    \qquad
    \abs{R_{m+1}(\theta)}\leq B\,\norm{\theta}_2^{m+1}.
\]
With $\abs{P_m(\theta)}\geq\kappa_m\norm{\theta}_2^m$ from
Eq.~\eqref{eq:ellipticity}, the reverse triangle inequality gives
\[
    \abs{\widehat H(\theta)}
    \;\geq\;
    \norm{\theta}_2^m\big(\kappa_m-B\norm{\theta}_2\big),
\]
and for $\norm{\theta}_2<\rho\leq\kappa_m/(2B)$, the right hand side of the inequality is larger than $\tfrac12\kappa_m\norm{\theta}_2^m$. This is exactly the expected behavior near the origin, where the leading term is significant compared to the remainder and concludes the proof for the first branch in Eq.~\eqref{eq:symbol-lower-piecewise}.

\emph{Away from the origin.}
The set $\{\theta:\norm{\theta}_2\geq\rho\}$ is compact and, since the zero at the origin is assumed to be the only zero $(Z=\{0\})$, the continuous function $\abs{\widehat H}$ attains a strictly positive minimum $ c_\star$. This proves the second branch in Eq.~\eqref{eq:symbol-lower-piecewise}.

\emph{Uniform bound.}
Now we find the uniform bound over the safe-band: For $\theta$ with $\norm{\theta}_2<\rho$ the first branch gives $\abs{\widehat H(\theta)}\geq\tfrac12\kappa_m\delta^m$; for $\norm{\theta}_2\geq\rho$ the second gives $c_\star$. The smaller of the two is the uniform bound.
\end{proof}

Given the lower bound for $H$, we lift up one of the key assumptions that the operator is elliptic. If the operator is non-elliptic, then the moment polynomial $P_m$ will vanish in some direction as it moves away from a zero, implying that we no longer have a guaranteed positive minimum rate $\kappa_m$ at which the symbol climbs out. However, we can use the \L ojasiewics inequality to prove the existence of the bound in non-elliptic conditions. Since the Fourier symbol $\widehat {H} $ for difference operators is trigonometric and therefore real-analytic, \L ojasiewicz inequality guarantees that $\abs{\widehat H}$ has a lower bound of the form $\abs{\widehat H(\theta)}\geq c\,\dist_2(\theta, Z)^{L}$ in the neighborhood of a zero. Here, $L\geq m$ and $c$ is a positive constant.

Finally, we lift up our assumption that the zero at the origin is the only zero. If $\widehat H$ vanishes away from the origin, the safe-band defined already has the notion of neighborhood from each of those zeroes in the form of $\dist_2(\theta, Z)$ for all elements in $Z$. The floor of all the piecewise lower bounds constructed as per Lemma \ref{lem:symbol-lower} and the \L ojasiewicz inequality becomes the minimum of the per-zero contributions, each governed by the local order at its zero.

\subsection{Success Probability}
We now use Lemma \ref{lem:symbol-lower} to develop the block encoding success probability guarantee. 
\begin{theorem}[Safe-band success probability]
For every normalized input $\ket{v}$ and every $\delta>0$,
\begin{equation}
    p_{\mathrm{succ}}(\ket{v})
    \;\geq\;
    \frac{g_m(\delta)^2}{\lambda^2}\,
    \eta_\delta(\ket{v}),
    \label{eq:safe-band-bound}
\end{equation}
and in the small-$\delta$ regime $\delta\leq(2c_\star/\kappa_m)^{1/m}$, where $g_m(\delta)=\tfrac12\kappa_m\delta^m$,
\begin{equation}
    p_{\mathrm{succ}}(\ket{v})
    \;\geq\;
    \frac{\kappa_m^2\,\delta^{2m}}{4\,\lambda^2}\,
    \eta_\delta(\ket{v}).
    \label{eq:safe-band-smalldelta}
\end{equation}
\label{thm:safe-band}
\end{theorem}
\begin{proof}
Retaining only the safe-band frequencies in the closed form Eq. \eqref{eq:succ-closedform} and applying $\abs{\widehat H(\theta)}\geq g_m(\delta)$ on $\Omega_\delta$ (Lemma \ref{lem:symbol-lower}),
\[
    p_{\mathrm{succ}}(\ket{v})
    \;\geq\;
    \frac{1}{\lambda^2}
    \sum_{\theta\in\Theta_N\cap\,\Omega_\delta}
    \abs{\widehat H(\theta)}^2\abs{\tilde v(\theta)}^2
    \;\geq\;
    \frac{g_m(\delta)^2}{\lambda^2}
    \sum_{\theta\in\Theta_N\cap\,\Omega_\delta}\abs{\tilde v(\theta)}^2
    =
    \frac{g_m(\delta)^2}{\lambda^2}\,\eta_\delta(\ket{v}),
\]
which gives Eq.~\eqref{eq:safe-band-bound}. The small-$\delta$ form can be recovered by the substitution $g_m(\delta)=\tfrac12\kappa_m\delta^m$.
\end{proof}
The bound derived in Eq.~\eqref{eq:safe-band-smalldelta} is the main result of this section. It represents the lower bound of $P_{succ}$ factorizing the operator term $\kappa_m^2\delta^{2m}$, the cost term $\lambda^{-2}$, and the input term $\eta_\delta$. The operator term is read off from moments alone, where the moments appear as exponents over the $\delta$ and the ellipticity constant factor.  The exponent $\delta^{2m}$ is the quantitative signature of the moment order: a higher-order stencil vanishes more smoothly at its zero set and thus incurs a steeper penalty as the input approaches the zero set. This clean separation of operator, cost, and input is invoked in Section \ref{sec:applications} for different families. We use the next section to further discuss the usefulness of Theorem \ref{thm:safe-band} and motivate its significance with an example.

\subsection{A class guarantee from moment data} 
\label{subsec:class-guarantee}

The block encoding circuit in itself does not need any knowledge of input for construction and is entirely determined by stencil coefficients and the offset-set. Moreover, for any given input, the exact success probability can be computed using Eq. \eqref{eq:succ-closedform}, and we observe that the moment-based lower bound is conservative relative, and we observe that the moment-based lower bound is conservative in relation to it (i.e., it never exceeds the exact value). Theorem \ref{thm:safe-band}, therefore, is not intended to provide sharper bounds for the probability of success for a per-input calculation but rather serves as a bridge from local moment information into an explicit guarantee for an entire class of inputs.

The bound makes the moment data relate to quantum cost. Eq.~\eqref{eq:safe-band-smalldelta} introduces the moment data directly into a success-probability guarantee through the scaling $\delta^{2m}$ and the ellipticity constant $\kappa_m$. The moment order thus controls a quantum cost quantity at the operator level, rather than appearing only implicitly after the full symbol is evaluated on a chosen input. This is the third component of the moment correspondence, and it holds uniformly in the grid size.

The bound in Eq.~\eqref{eq:safe-band-smalldelta} is nearly optimal over the input class it certifies. Among all normalized inputs satisfying
\[
    \eta_\delta(\ket{v}) \geq \eta_0,
\]
consider an input that places spectral mass $\eta_0$ on a safe-band frequency at which the symbol magnitude is smallest, and the remaining mass $1-\eta_0$ on the kernel mode $\theta=0$. Its exact success probability is
\[
    \frac{m_N(\delta)^2\,\eta_0}{\lambda^2},
\]
where
\[
    m_N(\delta)
    :=
    \min_{\substack{\theta\in\Theta_N\\
    \operatorname{dist}_2(\theta,Z)\geq\delta}}
    \left|\widehat H(\theta)\right|
\]
is the exact minimum of the symbol magnitude over the safe-band.
Consequently,
\[
    \inf_{\ket{v}:\,\eta_\delta(\ket{v})\geq\eta_0}
    p_{\mathrm{succ}}(\ket{v})
    =
    \frac{m_N(\delta)^2\,\eta_0}{\lambda^2}.
\]
Thus, the dependence on the input class, represented entirely by $\eta_0$, is sharp. The only conservatism in Eq.~\eqref{eq:safe-band-smalldelta} arises from replacing the exact quantity $m_N(\delta)$ with its moment-based lower bound, $g_m(\delta)$.

Finally, we present an example that demonstrates the usefulness of the bound and motivates its use for estimating block encoding cost. For operators that share a zero set, a normalized convention, and a safe-band regime, our class guarantee provides insights into their block encoding comparison. Consider two stencils that both discretize the $\partial_x^2$ operator. The standard second order differential stencil has the coefficients $\tfrac{1}{4}(1,-2,1)$ and the fourth-order-accurate refinement 
\begin{equation} 
    \tfrac{1}{64}\,(-1,\,16,\,-30,\,16,\,-1). 
\end{equation}
Both stencils are normalized, so $\norm{\widehat H}_\infty = 1$, the moment order is $m=2$, and we observe that $\lambda = \norm{\widehat H}_\infty = 1$, which makes the representation subnormalization optimal. The phases align at $\theta=\pi$, and therefore the operator-norm chain in Eq.~\eqref{eq:opnorm-chain} closes, thereby certifying its optimality (Proposition \ref{prop:opnorm}). This class guarantee, however, separates them. The symbols are the following.
\begin{align} 
    \widehat H_2(\theta) &= -\sin^2\!\Big(\tfrac{\theta}{2}\Big), \\ \widehat H_4(\theta) &= -\tfrac{1}{4}\,\sin^2\!\Big(\tfrac{\theta}{2}\Big) \Big(3 + \sin^2\!\Big(\tfrac{\theta}{2}\Big)\Big)
\end{align} 
Both symbols vanish to second order at the origin, but the leading coefficients give different ellipticity constants
\begin{equation} 
    \kappa_2^{(2)} = \tfrac{1}{4}, \qquad \kappa_2^{(4)} = \tfrac{3}{16}.
\end{equation}
Since the two stencils share $m$, $\lambda$, (for a fixed $\delta$) the safe-band and and $\mathcal{V}_{\delta,\eta_0} := \{\ket{v} : \eta_\delta(\ket v)\ge\eta_0\}$ their certified floors,
\begin{equation}
    \inf_{\ket{v}\in\mathcal{V}_{\delta,\eta_0}} p_{\mathrm{succ}}(\ket{v}) \;\geq\; \frac{\kappa_m^2\,\delta^{2m}}{4\lambda^2}\,\eta_0,
\end{equation}
have the following ratio
\begin{equation} 
    \frac{p_{\mathrm{floor}}^{(4)}}{p_{\mathrm{floor}}^{(2)}} = \bigg(\frac{\kappa_2^{(4)}}{\kappa_2^{(2)}}\bigg)^2 = \bigg(\frac{3/16}{1/4}\bigg)^2 = \frac{9}{16}.
\end{equation} 
Throughout the common small-$\delta$ regime, the fourth-order stencil achieves a certified floor $43.75\%$ smaller than that of the standard stencil, at identical $\lambda$, $\delta$, and $\eta_0$. This is a direct implication for their probability of success and repeat-until-success post-selection. The comparison exposes a genuine tension between classical accuracy and postselection performance under coefficient normalization: the fourth-order stencil eliminates the second-order truncation error (its fourth moment vanishes, giving error $O(h^4)$ rather than $O(h^2)$), but its smaller normalized leading coefficient lowers the certified success floor. Furthermore, representation has significant applications in stencil design and optimization, which are discussed as potential future work in Section \ref{sec:discussion}.

\section{Applications}\label{sec:applications}
In this section, we utilize the moment framework detailed above to analyze two classes of PDE operators: symmetric even-order stencils and asymmetric biased advection-diffusion stencils.

\subsection{Symmetric even-order stencils}\label{subsec:app-symmetric}
The first family comprises stencils whose coefficients are symmetric under reflection $c_{-t} = c_t$. Some of the most widely used operators have this property, including, but not limited to, the Laplacian, the biharmonic, and their higher-order polynomial analogs. We start by discussing properties shared by the family, then work out the  Laplacian and biharmonic operators as examples.

Applying the coefficient symmetry conditions, each offset $t$ has a matching $-t$ in $\mu_\alpha = \sum_t c_t\,t^\alpha$. Imposing the symmetry conditions $c_{-t} = c_t$,
\begin{equation}
    c_t\,t^\alpha + c_{-t}\,(-t)^\alpha
    = c_t\,t^\alpha\big(1 + (-1)^{\abs{\alpha}}\big).
    \label{eq:sym-parity}
\end{equation}
This vanishes whenever $\abs{\alpha}$ is odd, and therefore every odd-order moment of a symmetric stencil is $0$, thus the moment order is even. The next consequence of Eq. \eqref{eq:sym-parity} is that the symbols for symmetric stencils are real. Since, $\widehat H(\theta) = \sum_t c_t e^{-\ii\theta\cdot t}$, pairing up the symmetric coefficients collapses the $\widehat H(\theta)$ to a sum of cosines. By symmetry, the symbol is real and has even moments. For each operator treated below, we verify directly that their zero set is $Z=\{0\}$ and that they are elliptic. Therefore, the hypotheses of Lemma \ref{lem:symbol-lower} and Theorem \ref{thm:safe-band} hold for the entire family, with the even-moment order $m$ entering as the vanishing exponent.

Furthermore, we make three more key observations based on the fact that $m$ is restricted to even integers. First, by Corollary \ref{cor:recovery} the recovered continuum operator is the order-$m$ derivative $D_m = \tfrac{\mu_m}{m!}\,\partial_x^m$  (the Laplacian at $m=2$, the biharmonic $\partial_x^4$ at $m=4$). Second, by Theorem \ref{thm:symbol-expansion}, the symbol of the stencil vanishes up to order $m$ at the origin $\widehat H(\theta) = (-\ii)^m\tfrac{\mu_m}{m!}\theta^m + O(\theta^{m+1})$. Finally, by Theorem \ref{thm:safe-band}, the success probability scales as $\Theta(h^{2m})$ on band-restricted inputs. It is evident that an increase in the moment order, $m$, improves the derivative's approximation order, increases the depth of the zero at the origin, and steepens the asymptotic success probability near the zero set. The moment order is the sole controlling factor affecting all three key block-encoding parameters of the stencil.

\subsubsection{Laplacian $(m=2)$}. 
We now look at the specific case of the Laplacian, which has a moment order of $m=2$. The discrete Laplacian is one of the most widely used operators in the finite-difference literature \cite{oberman2013finite, huang2014numerical, duo2018novel, duncan2003accuracy} and has block-encoding circuits that produce optimal subnormalization ($\lambda=1$) \cite{sturm2025efficient}. In this section, we demonstrate that our proposed framework reproduces these optimal circuits and, furthermore, that it explains the empirical evidence presented in the literature.

\paragraph{Moments and Moment Polynomial:} We start by evaluating the simplest case of a 1-$D$ Laplacian, which has the finite difference coefficients
\begin{equation}
    c_{-1} = \tfrac14,\qquad c_0 = -\tfrac12,\qquad c_{+1} = \tfrac14,
    \label{eq:lap-stencil}
\end{equation}
which are the normalized coefficients for a second-difference stencil on $T = \{-1,0,1\}$. The zero-th and first-order moments vanish $\mu_0 = \tfrac14 - \tfrac12 + \tfrac14 = 0$ and
$\mu_1 = -\tfrac14 + \tfrac14 = 0$, while
\begin{equation}
    \mu_2 = \sum_{t}c_t\,t^2 = \tfrac14(1) + \tfrac14(1) = \tfrac12 \neq 0,
    \label{eq:lap-mu2}
\end{equation}
which means that the moment order is $m = 2$. By Corollary \ref{cor:recovery} the stencil
approximates the second-derivative operator $D_2 = \tfrac{\mu_2}{2!}\partial_x^2 = \tfrac14\partial_x^2$. The moment polynomial for the Laplacian then is $P_2(\xi) = \tfrac{\mu_2}{2!}\xi^2 = \tfrac14\xi^2$. Since the polynomial is constant on the unit sphere, $\abs{\xi} = 1$, the operator is elliptic and isotropic with
\begin{equation}
    \kappa_2 = \min_{\abs{\xi}=1}\abs{P_2(\xi)} = \tfrac14 .
    \label{eq:lap-kappa}
\end{equation}
We can lift this up to higher dimensions, trivially from the $1$-dimensional case by observing that the $d$-dimensional Laplacian is simply the axis-sum of one-dimensional second differences and is therefore separable. Then for $\lambda=1$ normalized $d$-dimensional Laplacian, $\mu_{2} = 1/(2d)$ for each of the basis directions and hence $P_2(\xi) = \|\xi\|_2^2/(4d)$, $\kappa_2 = 1/(4d)$. Thus, the safe-band floor carries the factor $1/d^2$, matching that calculated in Sturm-Schillo \cite{sturm2025efficient}.

\paragraph{The Symbol:} The Fourier symbol for the Laplacian as per Theorem \ref{thm:diagonalization} is 
\begin{equation}
    \widehat H(\theta)
    = \tfrac14 e^{-\ii\theta} - \tfrac12 + \tfrac14 e^{\ii\theta}
    = \tfrac12(\cos\theta - 1).
    \label{eq:lap-symbol}
\end{equation}
Observing the symbol, it has only one zero at $\theta = 0$, so the set $Z = \{0\}$ and the hypotheses of Lemma \ref{lem:symbol-lower} hold. This will enable us to calculate explicit constants for the safe-band. The supremum of the symbol is attained at
$\theta = \pi$,
\begin{equation}
    \norm{\widehat H}_\infty = \tfrac12\abs{\cos\pi - 1} = 1,
    \label{eq:lap-norm}
\end{equation}
and the subnormalization of the normalized stencil is $\lambda = \sum_t\abs{c_t} = \tfrac14 + \tfrac12 + \tfrac14 = 1$. The operator norm  (Eq.~\eqref{eq:opnorm-chain}) is then
\begin{equation}
    \norm{A_h}_{\mathrm{op}} = \norm{\widehat H}_\infty = \lambda = 1 .
    \label{eq:lap-optimal}
\end{equation}
By our Proposition \ref{prop:opnorm}, since the gap between the supremum and the subnormalization is $0$, this LCU construction of the operator is subnormalization-optimal. This implies that there exists no $(\lambda',m_a,0)$ block encoding of this operator with $\lambda' < 1$. We have recovered the same optimality from \cite{sturm2025efficient} through a separate general operator-norm relation. Our framework not only matches their subnormalization but also certifies that the value is optimal by showing that the symbol saturates the $\ell^1$ bound, which means that all stencil terms align in phase at $\theta = \pi$.

\paragraph{Observations regarding $P_{succ}$:} We now calculate the closed-form success probability using Theorem \ref{thm:success-closedform} on a single Fourier mode $\ket{\chi_\theta}$:
\begin{equation}
    p_{\mathrm{succ}}(\ket{\chi_\theta})
    = \abs{\widehat H(\theta)}^2
    = \tfrac14(1 - \cos\theta)^2
    \;\xrightarrow[\theta\to0]{}\;
    \frac{\theta^4}{16}.
    \label{eq:lap-psucc-mode}
\end{equation}
where $\lambda$ has been dropped as $\lambda=1$ and $\ket{\chi_\theta}$ is normalized. A physical mode of integer wavenumber $k$ is at the grid frequency $\theta = 2\pi k h$ (Definition \ref{def:brillouin}), and therefore the success probability in terms of $k$, as $h \to 0 $ is
\begin{equation}
    p_{\mathrm{succ}} \;\xrightarrow[h\to0]{}\; \frac{(2\pi k)^4}{16}\,h^4
    = C_k\,h^4,
    \qquad
    C_k = \frac{(2\pi k)^4}{16}.
    \label{eq:lap-Ck}
\end{equation}
Sturm and Schillo in \cite{sturm2025efficient} consider two one-dimensional test functions $ v_1(x) = \sin(2\pi x)$ and $v_2(x) = \cos(6\pi x)$ for which they calculate the effect of the Laplacian $ L_1v_1 = -(2\pi)^2 v_1$ and $L_1v_2 = -(6\pi)^2 v_2$. Their analysis produces the success probability constants $C_1 = \frac{(2\pi)^4}{16} = \pi^4$ and $C_2 = \frac{(6\pi)^4}{16} = 81\pi^4$ respectively. 

Our moment framework recovers the same constants that correspond exactly to the physical wavenumbers $k=1$ and $k=3$ in Eq.~\eqref{eq:lap-Ck}, producing  $C_{k=1} = \pi^4$ and $C_{k=3} = 81\pi^4$. The moment framework not only produces these constants but also identifies their origins, which are tied back to moments. The Laplacian symbol has an order-$2$ zero at the origin,
\begin{equation}
    \widehat H(\theta)
    =
    -\frac{\theta^2}{4}
    +
    O(\theta^4),
    \qquad \theta \to 0.
    \label{eq:lap-symbol-local}
\end{equation}
This order is exactly the moment order $m=2$, and the success probability of any fixed physical mode scales as $\Theta(h^{2m})$. For the Fourier mode $\ket{\chi_{k,h}}$ at integer wavenumber $k$, the prefactor is explicit: $p_{\mathrm{succ}}(\ket{\chi_{k,h}}) \sim \tfrac{(2\pi k)^4}{16}\, h^4$.
Thus, the ratio of the two test functions is

\begin{equation}
    \frac{p_{\mathrm{succ}}(\ket{\chi_{3,h}})}
         {p_{\mathrm{succ}}(\ket{\chi_{1,h}})}
    \xrightarrow[h\to0]{}
    \frac{C_{3}}{C_{1}}
    =
    \left(\frac{3}{1}\right)^{2m}
    =
    3^4
    =
    81.
    \label{eq:lap-frequency-ratio}
\end{equation}
This shows that the higher success probability of the more oscillatory input is not a feature of the Laplacian operator but rather a special $m=2$ case of the general moment-order law. An order-$m$ stencil suppresses low-frequency Fourier modes by a factor of $h^m$ at the amplitude level and, hence, by a factor of $h^{2m}$ at the success-probability level.

\paragraph{Safe-Band Analysis:} We apply our Theorem \ref{thm:safe-band} here with the parameters $\lambda = 1$, $m = 2$, $\kappa_2 = \frac14$.
\begin{equation} 
    p_{\mathrm{succ}} \;\geq\; \frac{\kappa_2^2\,\delta^4}{4}
    = \frac{\delta^4}{64}
    = \frac{(2\pi k_{\min})^4}{64}\,h^4 .
    \label{eq:lap-safeband}
\end{equation}
for band restricted inputs supported at wavenumbers $\abs{k} \geq k_{\min}$, where $\delta = 2\pi k_{\min}h$. Here, the $P_{succ}$ scaling is smaller by a constant factor of $4$. The exact closed form thus reproduces the reported exact constants, while the safe-band theorem trades this constant factor for a guarantee uniform over the input class as explained in Section \ref{subsec:class-guarantee}.

\subsubsection{Biharmonic $m=4$}

Now, we examine a different instance of the same even-moment order family. In this example, we show that the biharmonic operator with moment order $m=4$ demonstrates that the family characteristics are fully encompassed by the moment order. The normalized 
one-dimensional case for the stencil has the coefficients
\begin{equation}
    c_{\pm 2} = \tfrac{1}{16},\quad
    c_{\pm 1} = -\tfrac{4}{16},\quad
    c_0 = \tfrac{6}{16},
    \label{eq:biharm-stencil}
\end{equation}
defined on the set $T = \{-2,-1,0,1,2\}$. Again, as with the Laplacian, the odd-order moments vanish due to symmetry, and by moment calculation, $\mu_0 = \mu_2 = 0$ and $\mu_4 = \tfrac32 \neq 0$, so the moment order is $m = 4$. The stencil, then, approximates the fourth order derivative $D_4 = \tfrac{\mu_4}{4!}\partial_x^4 = \tfrac{1}{16}\partial_x^4$ and therefore has the moment polynomial $P_4(\xi) = \tfrac{\mu_4}{4!}\xi^4 = \tfrac{1}{16}\xi^4$ which is elliptic and isotropic with $\kappa_4 = \mu_4/4! = \tfrac{1}{16}$. The symbol is 
\begin{equation}
    \widehat H(\theta)
    = \tfrac{1}{16}\big(2\cos 2\theta - 8\cos\theta + 6\big)
    = \tfrac14(\cos\theta - 1)^2,
    \label{eq:biharm-symbol}
\end{equation}
which is real and has a 
$\;4^{th}$ order zero at the origin. Again, observing the symbol, it is evident that the supremum is at $\theta=\pi$, giving $\norm{\widehat H}_\infty = 1$, while the block encoding normalization constant $\lambda = \sum_t\abs{c_t} = (1 + 4 + 6 + 4 + 1)/16 = 1$. This means the gap is closed, and therefore the stencil block encoding design is sub-normalization-optimal:
\begin{equation}
    \norm{A_h}_{\mathrm{op}} = \norm{\widehat H}_\infty = \lambda = 1.
    \label{eq:biharm-optimal}
\end{equation}
Finally, the success probability has the same asymptotic behavior once we plug in the appropriate value for $m$. Since $\widehat H(\theta) = \theta^4/16 + O(\theta^6)$ near the origin,
\begin{equation}
    p_{\mathrm{succ}}(\ket{\chi_{k,h}})
    \;\xrightarrow[h\to0]{}\;
    \frac{(2\pi k)^8}{256}\,h^8
    = \Theta(h^{2m}), \qquad m = 4.
    \label{eq:biharm-psucc}
\end{equation}
The biharmonic is also the separable axis-sum extension $\sum_i \partial_{x_i}^4$, as with the Laplacian, and therefore no additional techniques are required to handle the $d$-dimensional cases. These examples illustrate that any symmetric even-order stencils form a family of block encodings whose parameters are determined solely by the moment order $m$, and that they are the same constructions read at potentially different values of $m$. 

\subsection{Asymmetric Advection-Diffusion Stencils}
\label{subsec:app-advection}
The next family showcases the flexibility and capability of the moment framework. Asymmetric advection-diffusion stencils drop the reflection symmetry. Any odd-order derivative must be approximated with an odd moment order, and any stencil with an odd moment order must be asymmetric. This family provides a structural framework for a wide class of odd-derivative operators, including transport equations \cite{simpson2011corrected, rood1987numerical}. To our knowledge, this family does not have an explicit spatial block encoding in the quantum literature, since the prior symmetric Laplacian method does not naturally extend to it.

The moment framework provides lower bounds on the probability of successful implementation; however, there are three complications arising from the block encoding: (i) complex symbol, (ii) loss of ellipticity for the polynomial, (iii) multiple zeros exactly as those discussed in Section \ref{subsec:lowerbound}.

Specifically, we consider the operator $a\,\partial_x + b\,\partial_x^2$, where $ a$ denotes the advection speed and $ b$ the diffusion coefficient. For notational simplicity, we use the dimensionless convention in which the mesh spacing is $h=1$; restoring the grid spacing replaces $a$ and $b$ in the stencil by $a/h$ and $b/h^2$, respectively. When discretizing this operator on $T = \{-1,0,1\}$, the coefficients from the central difference equations are
\begin{equation}
    c_{+1} = \tfrac{a}{2} + b,
    \qquad
    c_{0} = -2b,
    \qquad
    c_{-1} = -\tfrac{a}{2} + b .
    \label{eq:adv-stencil}
\end{equation}

\begin{remark}[Coefficient reflection for odd-order targets]
Under the ket-action convention $S_t\ket{x}=\ket{x+t}$, the operator built from the coefficients in Eq. (49) realizes $-a\,\partial_x + b\,\partial_x^2$ at leading order (Corollary \ref{cor:recovery} with $(-\partial)^\alpha$). To realize the target $a\,\partial_x + b\,\partial_x^2$ we take the reflected coefficients $\tilde c_t := c_{-t}$, i.e.\ $\tilde c_{+1} = b - \tfrac a2$, $\tilde c_{-1} = b + \tfrac a2$, $\tilde c_0 = -2b$. Reflection flips the sign of every odd moment ($\mu_1 \mapsto -\mu_1$) and leaves $\lvert\mu_\alpha\rvert$, $\lambda$, $\lvert\widehat H\rvert$, the zero set, and every bound of this subsection unchanged. We present all results for the representative coefficients in Eq. (49) with $\mu_1 = a > 0$.
\end{remark}

Its moments are $\mu_0 = 0$, $\mu_1 = \sum_t c_t\,t = a \neq 0$, and $\mu_2 = 2b$ where $b\geq0$. When $a\neq0$, $m=1$ and therefore the stencil approximates the first-derivative part of the operator $D_1 = \tfrac{\mu_1}{1!}\partial_x = a\,\partial_x$ (Corollary \ref{cor:recovery}). The consequences of the odd-order moment follow directly in calculating the symbol.
\begin{align}
    \widehat H(\theta) =& (a/2+b)e^{-i\theta} - 2b + (-a/2+b)e^{i\theta}\\
    =&\underbrace{2b\;(\cos\theta - 1)}_{\text{diffusion, real, even}}
    \;-\; \ii\,\underbrace{a\;\sin\theta.}_{\text{advection, imag, odd}}
    \label{eq:adv-symbol}
\end{align}
Thus, the symbol is complex, with the asymmetric advection term producing an imaginary odd component. Near the origin, 
\begin{equation} 
\widehat H(\theta) = -\ii a\theta - b\theta^2 + O(\theta^3), \end{equation} 
so the leading imaginary term $-\ii a\theta$ represents the first-derivative approximation, while the second moment determines the quadratic diffusive term. A separate complication appears when the one-dimensional derivative is embedded as a directional operator in multiple dimensions.

The next consequence of asymmetry is that the leading polynomial is directional. At the origin, the leading polynomial $P(\xi)=\mu_1 \xi= a\xi$ is one-dimensional. For this one-dimensional case, the unit sphere is simply $\{-1, +1\}$ and the ellipticity is preserved since $\abs{P(\pm1)}=a\neq0$. However, for multi-dimensional cases, the symbol does not retain its ellipticity. For example, consider a two-dimensional case where we discretize the operator $a\partial_{x_1} + 0\cdot\partial_{x_2}$. On the unit circle, $\xi = (\cos\phi, \; \sin\phi)$, the leading coefficient is $a$, and the polynomial $P = a\;\cos\phi$ vanishes at $\phi = \pi/2$ throughout the plane, making it vanish on the perpendicular $x_2$ axis as well. Thus, $\kappa_1 = \min_{\norm{\xi}=1}\abs{P_1(\xi)} = 0$ and the operator is not elliptic in multiple dimensions.

Finally, the operator symbol may carry additional zeros, which depend on the mixture of diffusion and advection. Starting from the symbol in Eq.~\eqref{eq:adv-symbol}, we observe that the zeros of the operators are precisely those that make both the real and imaginary parts vanish simultaneously. The imaginary part $-a\sin\theta$ vanishes at $\theta \in \{0, \pi\}$ whereas the real part $2b(\cos\theta - 1)$ vanishes only at $\theta = 0$. This means that,
\begin{equation}
    Z =
    \begin{cases}
        \{0,\ \pi\}, & b = 0 \quad\text{(pure advection)},\\[1ex]
        \{0\}, & b > 0 \quad\text{(advection--diffusion)},
    \end{cases}
    \label{eq:adv-zeroset}
\end{equation}
implying the zero set is determined by whether diffusion is present. For a pure advection operator, there is a zero at the Brillouin edge $\theta=\pi$, which is the highest frequency visible to the grid, and naturally the central-difference advection annihilates it. Adding the diffusion lifts this zero up, regularizing the family from two zeros to a single zero, similar to the symmetric family. Therefore, the existence of diffusion determines how we analyze the block encoding for this family. The diffusion coefficient switches between the multi-zero and single-zero cases, and each case calls for the relaxations discussed in Section \ref{subsec:lowerbound}.

Since the symbol for this family admits a $\sin\theta$ and $1-\cos\theta$ representation, we can obtain global lower bounds without relying on compactness. We still show that these bounds are derived from the moment framework. We start with a pure advection symbol and then move on to advection-diffusion.

\begin{lemma}[Pure Advection] 
\label{lem:adv-pure}
For pure advection, $b = 0$ and $a > 0$, so that $\widehat H(\theta) = -\ii\,\mu_1\sin\theta$ with $Z = \{0,\pi\}$ from Eq.~\eqref{eq:adv-zeroset}. Then for every $\theta \in [-\pi,+\pi]$, $\widehat H$ is bounded below as
\begin{equation}
    \abs{\widehat H(\theta)}
    \;\geq\;
    \frac{2\mu_1}{\pi}\,\dist(\theta, Z).
    \label{eq:adv-pure-bound}
\end{equation}
\end{lemma}
\begin{proof}
Remember that from the moment data $\mu_1 = a$ and $\mu_2 = 2b$. On the interval $\theta \in [-\pi,\pi]$, the distance of any theta to the zero set $Z = \{0,\pi\}$ is at most $\pi/2$:
\[
d := \dist(\theta, Z) = \min(\abs{\theta},\,\pi - \abs{\theta}) \in [0, \pi/2].
\]
Since the absolute value of $\sin\theta$ depends on the distance between $\theta$ and $\pi$, reflecting $\theta$ about its nearest zero yields $\abs{\sin\;\theta}=\sin\;d$. We apply Jordan's inequality $\sin d \geq \tfrac{2}{\pi}d$ on $[0,\pi/2]$, which then yields $\abs{\widehat H(\theta)} = \mu_1\abs{\sin\theta} = \mu_1\sin d \geq\tfrac{2\mu_1}{\pi}\,d$. Furthermore, note that the Jordan inequality yields equality at $d = \pi/2$, corresponding to $\theta = \pm\pi/2$, and therefore the equality is sharp.
\end{proof}
This lemma shows that the pure advection symbol is globally bounded below, derived from the $m = 1$ statement of the framework. The linear growth $\dist(\theta, Z)^{1}$ is the $\delta^{m}$ law at $m = 1$, where, unlike the elliptic cases where $\kappa_m$ is used, the moment $\mu_1$ directly shows up as the rate. Now we extend this idea to the advection-diffusion case.

\begin{lemma}[Advection-Diffusion] 
\label{lem:adv-diff}
Let both $a, b>0$ and from Eq.~\eqref{eq:adv-zeroset}, $Z = \{0\}$. Then for every $\theta \in [-\pi,\pi]$,
\begin{equation}
    \abs{\widehat H(\theta)}
    \;\geq\;
    \frac{2}{\pi}\,\min(\mu_1,\,\mu_2)\,\abs{\theta}.
    \label{eq:adv-diff-bound}
\end{equation}
When $\mu_1 > \mu_2$ the bound is attained at $\theta = \pi$, where
$\abs{\widehat H(\pi)} = 2\mu_2$.    
\end{lemma}

\begin{proof}
Starting again from the symbol, Eq.~\eqref{eq:adv-symbol}, and squaring it,
\[
\abs{\widehat H(\theta)}^2 = \mu_2^2(1-\cos\theta)^2 + \mu_1^2\sin^2\theta.
\]
Substituting $u := 1 - \cos\theta \in [0,2]$ and $\sin^2\theta = (1-\cos\theta)(1+\cos\theta) = u(2-u)$ gives
\[
    \abs{\widehat H(\theta)}^2
    = \mu_2^2 u^2 + \mu_1^2 u(2-u)
    = u\big[(\mu_2^2 - \mu_1^2)\,u + 2\mu_1^2\big].
\]
Here, the polynomial inside the square brackets is linear in the interval $[0,2]$. This means that its minimum must be on one of the endpoints. Calculating, $2\mu_1^2$ at $u = 0$ and $2\mu_2^2$ at $u = 2$. Thus $\abs{\widehat H(\theta)}^2 \geq 2u\,\min(\mu_1^2, \mu_2^2)$. We use the trigonometric relation $1-\cos \theta=2\sin^{2}\theta/2$ and apply Jordan's inequality to $\theta/2$ again to arrive at $\abs{\widehat H(\theta)}^2 \geq \tfrac{4}{\pi^2}\min(\mu_1^2,\mu_2^2)\, \theta^2$. Taking the square root of both sides yields the final result. We further note that if $\mu_1 > \mu_2$ then $\tfrac{2}{\pi}\mu_2\cdot\pi = 2\mu_2$ as the symbol, $\widehat H(\pi) = -2\mu_2$, is real at $\theta=\pi$.
\end{proof}
The two lemmas provide explicit family-wise constants of the type anticipated in Section \ref{subsec:lowerbound}. In particular, the pure-advection lower bound is controlled by the first moment, whereas the advection-diffusion lower bound depends jointly on the first two moments through $\min(\mu_1,\mu_2)$. Now that we have the explicit lower bounds for the family, we can apply the safe-band theorem as per Section \ref{sec:safe-band-psuc}. For pure advection, $\lambda = \abs{\mu_1/2} + \abs{-\mu_1/2} = \mu_1$, and Lemma~\ref{lem:adv-pure} gives, for every $\delta \in (0, \pi/2]$,
\begin{equation}
    p_{\mathrm{succ}}(\ket{v})
    \;\geq\;
    \frac{4\mu_1^2}{\pi^2\lambda^2}\,\delta^2\,\eta_\delta(\ket{v})
    \;=\;
    \frac{4}{\pi^2}\,\delta^{2}\,\eta_\delta(\ket{v}),
    \label{eq:adv-pure-succ}
\end{equation}
with an explicit constant and the $\delta^{2m}$ law at $m = 1$. The spectral mass $\eta_\delta$ counts spectral mass away from both the zeros in the $Z$ set. Similarly, for advection-diffusion, the subnormalization originates from Eq.~\eqref{eq:adv-stencil},
\begin{equation}
    \lambda
    = \abs{\tfrac{\mu_1}{2} + \tfrac{\mu_2}{2}} + \mu_2 +
      \abs{\tfrac{\mu_2}{2} - \tfrac{\mu_1}{2}}
    = \mu_2 + \max(\mu_1, \mu_2),
    \label{eq:adv-lambda}
\end{equation}
and Lemma~\ref{lem:adv-diff} gives
\begin{equation}
    p_{\mathrm{succ}}(\ket{v})
    \;\geq\;
    \frac{4}{\pi^2}\,
    \frac{\min(\mu_1,\mu_2)^2}{\big(\mu_2 + \max(\mu_1,\mu_2)\big)^2}\,
    \delta^{2}\,\eta_\delta(\ket{v}).
    \label{eq:adv-diff-succ}
\end{equation}
Every quantity in Eq. \eqref{eq:adv-pure-succ} and Eq. \eqref{eq:adv-diff-succ} is a function of the first two moments of the stencil, with explicit constants worked out, showcasing the reach and depth of this framework.

\section{Discussion}\label{sec:discussion}
In this work, we have shown that a single integer, $m$, representing the moment order of a finite-difference stencil, simultaneously determines the (i) continuum operator recovered, (ii) the symbol's vanishing structure via symbol expansion, and (iii) the block encoding cost through both the optimality verdict using the operator-norm chain and the success probability using $\delta^{2m}$. This framework sits on a unique middle ground, class-level constructions that certify no optimality, and optimal constructions that do not transfer across operators. Our framework pivots on a single structural parameter that yields a uniform block-encoding method across multiple operator families, while certifying when the operator-norm gap closes, thereby yielding an optimal subnormalization condition for the block-encoding design that subsumes the conditions derived individually for operator-specific block-encoding constructions. The optimal Laplacian construction put forward by Sturm and Schillo in \cite{sturm2025efficient} is recovered as the gap-closed instance of our theorem in Section \ref{subsec:app-symmetric}, while the same framework is used to find constructions of higher-order stencils (biharmonic), and the explicit constants have been calculated for the advection-diffusion family treated at $m=1$, for which no spatial block encoding exists.

There are several avenues for future research. In this work, we have only considered periodic boundary conditions, translation invariance, and constant coefficients. The optimality condition for subnormalization applies only to exact encodings, and the ancilla qubit counts and circuit depth are inherited directly from the prepare-select construction of LCUs, without optimization. Again, the class guarantee of Theorem \ref{thm:safe-band} is conditional on the input spectral mass inside the safe band. Although for this particular observation, the limiting factor is not the block encoding design but rather the fact that inputs with spectral mass concentrated near the zero set of the symbol are attenuated by the operator itself. Thus, the theorem states the honest price of an input-uniform statement. For non-elliptic or multi-zero symbols, relaxations of the results in Section~\ref{sec:safe-band-psuc} still guarantee the existence of a success-probability floor but do not state the constants explicitly.

Future work includes extending this framework to accommodate Neumann and Dirichlet boundary conditions, as well as variable-coefficient operators. This may be possible by connecting this work with the approaches developed in \cite{kharazi2025explicit, zecchi2026block, li2023efficient}. Another promising direction would be to carry this analysis further downstream, connect it to QSVT primitives, and analyze how the framework affects the localization of approximation errors throughout the pipeline. The safe-band factorization may also be useful beyond analysis, as it could serve as an objective for designing and optimizing stencils that retain a prescribed order of approximation while enabling more efficient block encodings.

%\section*{Acknowledgments}

\section*{Competing interests}
All authors declare no financial or non-financial competing interests. 

%\section*{Code and Data Availability}

\bibliography{manuscript}
\bibliographystyle{unsrt}
\end{document}